\documentclass[a4paper,twocolumn,11pt,accepted=2025-11-18]{quantumarticle}
\pdfoutput=1

\usepackage[utf8]{inputenc}
\usepackage[english]{babel}
\usepackage[T1]{fontenc}
\usepackage{amscd,amsmath,amssymb,amsthm,mathtools}
\usepackage{enumitem,graphicx,hyperref,xcolor}
\usepackage{tikz}

\usepackage[normalem]{ulem}
\usepackage{braket,cleveref,comment,csquotes,dsfont,ifthen,setspace}

\newtheorem{theorem}{Theorem}[section]
\newtheorem{lemma}[theorem]{Lemma}
\newtheorem{corollary}[theorem]{Corollary}
\newtheorem{definition}{Definition}[section]

\newcommand\numberthis{\addtocounter{equation}{1}\tag{\theequation}}
\newcommand\chris[1]{\textcolor{blue}{\textbf{#1}}}

\newcommand{\CZ}{\mathrm{CZ}}
\newcommand{\SWAP}{\mathrm{SWAP}}
\newcommand{\FSWAP}{\mathrm{FSWAP}}

\DeclareMathOperator{\tr}{\mathbf{Tr}}

\setlist[itemize]{leftmargin=3.5mm,itemsep=0.0mm}

\begin{document}

\title{Operator space fragmentation in perturbed Floquet-Clifford circuits}

\author{Marcell D. Kovács}
\affiliation{Department of Physics and Astronomy, University College London, United Kingdom}
\author{Christopher J. Turner}
\affiliation{Department of Physics and Astronomy, University College London, United Kingdom}
\author{Lluís Masanes}
\affiliation{Department of Computer Science, University College London, United Kingdom}
\affiliation{London Centre for Nanotechnology, University College London, United Kingdom}
\author{Arijeet Pal}
\affiliation{Department of Physics and Astronomy, University College London, United Kingdom}
\affiliation{London Centre for Nanotechnology, University College London, United Kingdom}

\maketitle

\begin{abstract}
Floquet quantum circuits are able to realise a wide range of non-equilibrium quantum states, exhibiting quantum chaos, topological order and localisation. The circuit based perspective has led to new regimes of many-body quantum dynamics with potential applications to quantum technologies.  In this work, we investigate the stability of operator localisation and the emergence of chaos in random Floquet-Clifford circuits subjected to unitary perturbations which drive them away from the Clifford limit. We construct a nearest-neighbour Clifford circuit with a brickwork pattern and study the effect of including disordered non-Clifford gates. The perturbations are uniformly sampled from single-qubit unitaries with probability $p$ on each qubit. We show that the interacting model exhibits strong localisation of operators for $0 \leq p < 1$ that is characterised by the fragmentation of operator space into disjoint sectors due to the appearance of \textit{wall} configurations. Such walls give rise to emergent local integrals of motion for the circuit that we construct exactly. We analytically establish the stability of localisation against generic perturbations and calculate the average length of operator spreading tunable by $p$. Although our circuit is not separable across any bi-partition, we further show that  the operator localisation leads to an entanglement bottleneck, where initially unentangled states remain weakly entangled across typical fragment boundaries. Finally, we study the spectral form factor (SFF) to characterise the chaotic properties of the operator fragments and spectral fluctuations as a probe of non-ergodicity. In the $p=1$ model, the emergence of a fragmentation time scale is found before random matrix theory sets in after which the SFF can be approximated by that of the circular unitary ensemble. Our work provides an explicit description of quantum phases in operator dynamics and circuit ergodicity which can be realised on current NISQ devices.
\end{abstract}

% % \documentclass[11pt, floatfix, superscriptaddress, aps, prb]{revtex4-2}
% % \documentclass[10pt, twocolumn, floatfix, superscriptaddress, aps, prb, raggedbottom]{revtex4-2}

% % \bibliographystyle{apsrev4-2}

\singlespacing

\section{Introduction}

\begin{figure}[tp]
    \centering
    \includegraphics[width=0.48\textwidth]{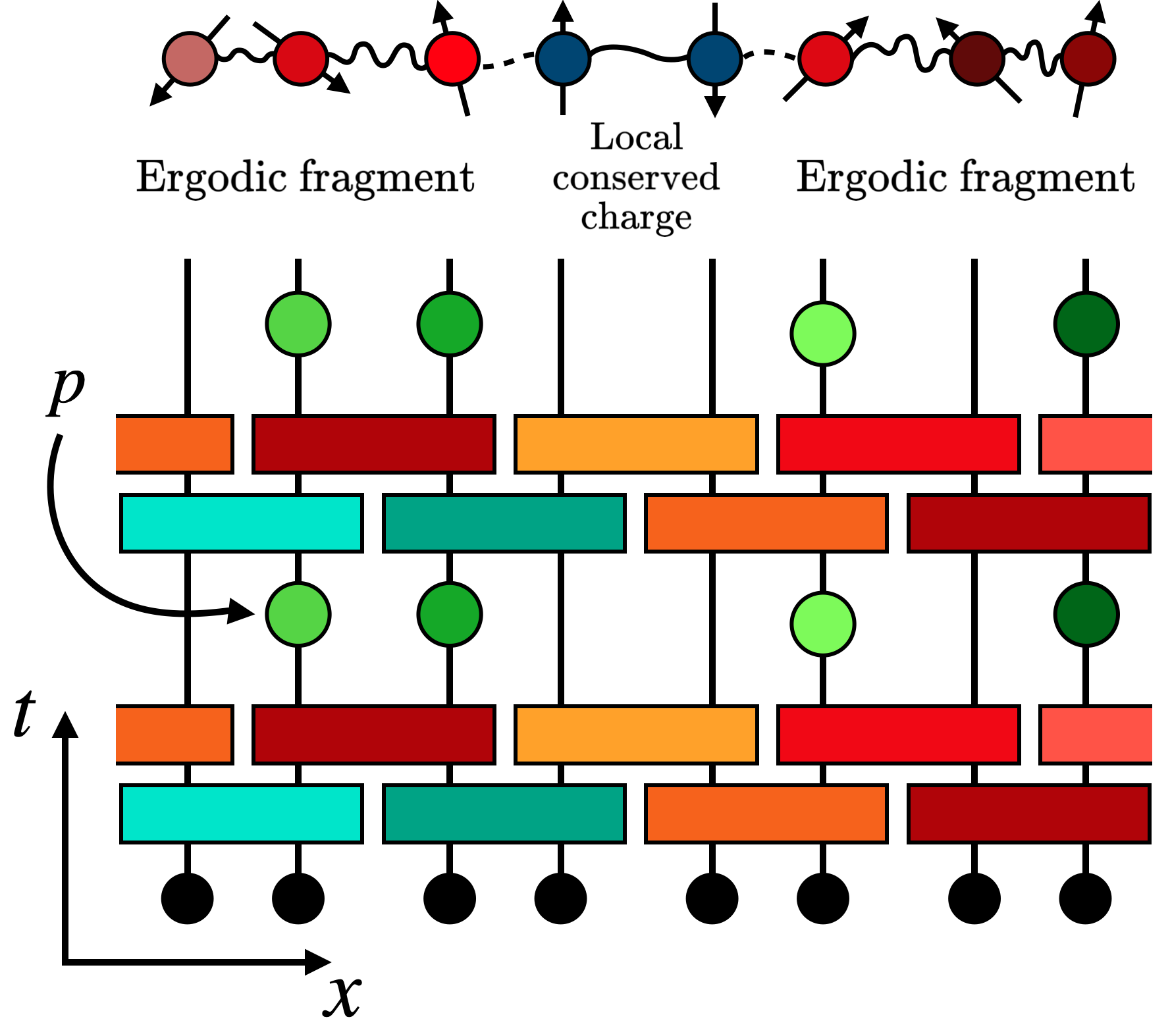}
    \caption{\small A segment of the brickwork Floquet circuit considered in this work. We define a unitary evolution operator according to Equation (\ref{eq:floquet_unitary}) on a one-dimensional qubit chain. The gates consist of randomly sampled entangling Clifford gates (rectangles) and random $\mathrm{SU}(2)$ rotation gates (coloured circles) applied stochastically with probability $p$. Same colours represent the same gate.}
    \label{fig:floquet_circuits}
\end{figure}

Probing the dynamics of quantum many-body systems is a challenging problem. Many-body systems far away from equilibrium can scramble information and generate complex patterns of entanglement which are hard to simulate using classical algorithms. Emergence of ergodicity and chaos provides a physical manifestation of the complexity of the quantum states~\cite{roberts2017chaos, Eisert2021CircuitComplexity, WenWei2022StateDesign, haferkamp2022complexity}. A circuit based simulator comprising of quantum gates provides a controlled physical platform where these complex states can be realized and tested against perturbations in the form of errors~\cite{Fisher2023, Hahn2024, Chao-Ming2020, Nahum2017, Nahum2018, Li2018, Lunt2021}. Understanding the robustness of complex entanglement can provide paths to quantum cryptography~\cite{Horodecki2009_RMP, Jennewein2000_PRL, yin2020qcrypt} and error correction~\cite{Bennett_QEC1996_PRA, Terhal_RMP2015}. The dynamics of quantum information can also shed light on the physical processes which lead to thermalization and provide new perspectives on the breaking of ergodicity~\cite{kaufman2016thermal, nandkishore2015MBL, Abanin2019_RMP}. 

An illuminating direction has been the study of toy models for ergodicity, for example, dual unitary circuits with quantum gates encoding a space-time symmetry, which serve as solvable, minimal, time-periodic models of quantum chaotic evolution~\cite{Bertini2019, Bertini2021, Claeys2020}. On the other hand, random brick-work circuits consisting of two-qubit gates are rigorously established to converge to a many-qubit Haar-ensemble, a provably ergodic case for unitary dynamics~\cite{Harrow2009, Brandao2016, Gross2007, Haferkamp2022, Huangjun2017, Chan2018MBQC}. They also provide simplified descriptions of operator growth in terms of classical stochastic processes~\cite{Nahum2018, CurtvK2018}. In contrast, non-ergodic quantum behaviour known as many-body localisation (MBL) occurring in many-body Hamiltonians with quenched disorder~\cite{Gornyi2005, Basko2006, Pal2010, Luitz2015} have been argued to exist in time-periodic (Floquet) circuits~\cite{Chandran2015, Christoph2018, Chan2021SpectralLyapunov, Farshi2022_1D, Farshi2022_2D}. They are characterised by slow growth of entanglement and provide a new regime for operating quantum circuits with potential use as a quantum memory. Therefore, quantum circuits provide a fertile playground for investigating the phenomena of many-body localisation and ergodicity in a new regime which is tractable.

Many-body localisation is characterised by the emergence of conservation laws, given by quasi-local operators, known as $l$-bits, that are not scrambled under time evolution while the growth of entanglement of pure quantum states is logarithmic in time~\cite{Serbyn2013, Huse2013b}. The existence of $l$-bits at asymptotically large disorder can be proven for a one-dimensional spin-$\frac{1}{2}$ chain~\cite{Imbrie2016}, while their susceptibility to resonances leading to an avalanche at weaker disorder is a topic of active investigation~\cite{DeRoeck2017, Luitz2017, Sels2022, Crowley2022b, Morningstar2022, Ha2023, Jeyaretnam2023_MBLresonance}. Understanding the stability of the localised subspaces incurs exponential computational costs~\cite{Wahl2017MBL-TN,  Doggen2018MBL, wahl20192DMBL,  doggen2021MBL-MPS} and requires analytic approximations which are hard to verify in the absence of an alternate proof. From this perspective, certain classes of quantum circuits can provide a simplified analytic structure for describing operator growth and where the approximations are concrete and clearly testable.

%Quantum circuits have emerged as a novel platform to study interacting dynamics. The analytical tractability of calculating entanglement properties, out-of-time-ordered correlators and the ability to rigorously show the emergence of chaotic random matrix ensembles have motivated a plethora of non-equillibrium quantum phases and phase transitions \cite{Fisher2023, Hahn2024, Chao-Ming2020, Nahum2017, Nahum2018, Li2018, Lunt2021}. 

%A particular direction has been the study of toy models for ergodicity, for example, through the construction of dual unitary circuits that have been used as minimal, non-random, models of quantum chaotic evolution \cite{Bertini2019, Bertini2021, Claeys2020}. On the other hand, random brickwork circuits made from two-qubit random gates under uniform measure rigorously established the convergence to a many-qubit Haar-ensemble showing an example of provable ergodicity in the case of unitaries \cite{Harrow2009, Brandao2016, Gross2007, Haferkamp2022, Huangjun2017}. In the adjacent direction, non-ergodic quantum behaviour known to occur in many-body Hamiltonians with quenched disorder \cite{Abanin2019} have also been suggested for time-periodic (Floquet) circuits which have allowed the adaptation of well-studied localisation phenomena to the circuit environment \cite{Nahum2018,Chandran2015, Fisher2023, Christoph2018, Farshi2022_1D, Farshi2022_2D}.

In this context, Clifford circuits have been a convenient tool for solvable models of non-trivial dynamical phases due to their analytical and simulable structure~\cite{haah2022topological, lavasani2021topological, Lunt2021, Sommers2023CQC, angelidi2023LRE, Richter2023_LRClifford}. For these circuits, both state and operator evolution can be simulated efficiently classically under the Gottesmann-Knill theorem using a binary phase space mapping technique \cite{Tolar2018, Aaronson2004}. Circumventing the entanglement barrier for a large class of quantum many-body systems allows the simulation of entangled phases of matter and their transitions \cite{Lunt2021, Li2024CliffordMonitor, Makki2024}. Although the classically simulable models provide a useful starting point to describe the phenomena, a generalised description requires understanding the role of non-Clifford perturbations~\cite{Bravyi2016nonClifford}. Recently,  the relevance of \enquote*{magic} as a  measure of non-Clifford resources has been highlighted for the study of  transitions in classes of chaotic systems and universal quantum computation~\cite{niroula2023MagicPT, bejan2023Magictrans, TarabungaPRX2023_Magic}. Since the Clifford group with the inclusion of even a single $T$-gate forms an (approximately) universal gateset, one expects that a disordered Clifford ensemble with random unitary perturbations can approach the properties of a circuit with Haar-distributed unitaries. Several results in the literature for models without spatial structure indicate so: a single $T$ gate is able to drive a spectral transition for Clifford circuits towards a random matrix ensemble and a system-size independent number of non-Clifford gates are sufficient to approximate the Haar-ensemble by asymptotically forming an approximate quantum $k$-design \cite{Gross2007, Zhou2020, Haferkamp2022, Hunter-Jones2019}. Originally developed for quantifying pseudo-randomness, $k$-designs for both random states and unitary ensembles have been investigated as model for the thermalisation dynamics of Floquet and quasi-periodically driven systems \cite{Pilatowsky-Cameo2023_CHSE_Fibonacci, Pilatowsky-Cameo2024_CHSE_CUE}. These approaches, however, do not consider the possibility of spatial localisation and its effect on ergodicity.

 Our work involve the study of localisation in Clifford circuits as the emergence of local conserved quantities (similar to MBL) and their instability to perturbations which drive the system towards ergodicity. Constructing local integrals of motion in disordered many-body Hamiltonians, while pivotal in understanding the phenomenology of MBL, is a difficult task both analytically and numerically \cite{Serbyn2013, Imbrie2016, Rademaker2016, Kulshreshtha2018, Thomson2018, Goihl2018}. Previous work has considered a semi-classical version of MBL in the circuit setting using a subset of Clifford gates generating local conservation laws~\cite{Chandran2015}, however, a more rigorous and complete analysis of the invariant subspaces occurring in random Clifford circuits and their relation to localisation phenomenon is so far lacking. Furthermore, open questions remain on the stability of localisation away from the Clifford limit and the dynamics of entanglement in this new regime. 

These considerations motivated us to study the interplay of Clifford dynamics in brickwork circuits with random non-Clifford perturbations as a toy model to understand the robustness of Floquet localisation and study the emergence of ergodicity in our model. This paper focuses on operator spreading in random circuits in one-dimension for which the coexistence of localisation and Pauli mixing was rigorously established for periodic unitaries in the Clifford group and soon after ruled our for two-dimensions \cite{Farshi2022_1D, Farshi2022_2D}. Our work not only extends these previous results but sheds light on the random Floquet model in the fully interacting regime. Clifford dynamics is associated with quasi-free particles rather than a strongly correlated system due to the classically simulable phase-space dynamics \cite{Farshi2022_1D}.

On the analytical side, we formalise the notion of operator localisation in Clifford circuits and construct irreducible gate configurations which generate arrested operator spreading (so-called $k$-walls). This generalises earlier approaches \cite{Chandran2015, Farshi2022_1D, Farshi2022_2D} by elucidating the necessary and sufficient structure of these gates and their induced dynamics as well as analytically calculating the probability of their occurrence in a brickwork circuit ensemble. From this, invariant subspaces (spatially localised fragments) of the circuit unitary can be identified which naturally leads up to the stability against non-Clifford perturbations and the estimation of localisation length in our model without requiring numerical simulations. 

An advantage of using Clifford circuits as a starting point is that conservation laws can be constructed efficiently and their efficient simulability allows their detection with polynomial computational cost. Although approximate simulation techniques exist for including sparse $T$-gate perturbations in Clifford circuits, these still suffer from an exponential computational cost with increasing circuit depth \cite{Bravy2019}. As a result, we complement our analytical results with small-scale exact numerical calculations of entanglement and chaos measures such as the spectral form factor -- a well-established probe in random matrix theory~\cite{Brezin1997, Haake1999, Chan2018}. 

We refine the previous understanding of Clifford localisation as the fragmentation of operator space into invariant subspaces corresponding to spatially localised subsystems. Fragmentation (in both state and operator space) has previously been studied in the context of open quantum systems~\cite{Essler2020OSF, YahuiLi2023_OpenHSF, Paszko2024_OpenSPT} while our work highlights the role of the phenomena in unitary circuit dynamics. For a comprehensive review on Hilbert space fragmentation, the Reader is referred to~\cite{moudgalya2022Scars_HSF}. We also show that wall configurations generating fragmented subspaces are typically associated with local conservation laws. Our main novel contribution is to show the stability of localisation against generic, disordered, unitary perturbations. We also consider uniform perturbations leading to operator percolation and chaotic evolution with connections to one-dimensional percolation and reproducing some aspects of Anderson localisation. In our model, the qubit chain breaks up into spatially localised fragments which, within themselves, evolve under an approximate Haar ensemble while the fragments share limited entanglement across their boundaries. This serves as a toy (yet not semi-classical) model of many-body localisation where local conservation laws provably occur without fine-tuning of the circuit elements. We note, however, that our model does not exhibit the $l$-bits conventionally associated with the MBL phase, even in the localised regime. This is due to the ergodic circuit regions which scramble operators as soon as they escape the unperturbed $k$-walls hosting local charges. The conserved quantities are strictly local and separate large ergodic regions. We also rule out the stability of localisation in the uniformly perturbed regime where each qubit has generic non-Clifford evolution due to the destabilisation of walls. This contrasts earlier results presented in Ref. \cite{Chandran2015} which we attribute to their finely-tuned circuit construction. Although our results focus on one-dimensional qubit circuits, we suggest geometries where localisation could be engineered for higher dimensional lattices, and also generalisations to higher local Hilbert space dimensions.

The paper is structured as follows. In the next section, we review the group theoretic formalism of Clifford and Pauli groups as the foundation of our analytical investigations. Then, we define our Floquet model and introduce the spreading of operators in Section \ref{sec:floquet_model}. Section \ref{sec:wall_configs} is dedicated to localisation in Floquet Clifford dynamics formulated in terms of the constraints on circuit elements which generate fragmented evolution. \Cref{lemma:conserved_pauli} establishes the existence of closed subspaces of the evolution which we utilise to understand the stability of localisation. We study the dynamics of entanglement in pure state evolution in Section \ref{sec:entanglement} to show the bounded entropy shared by localised fragments in the circuit using the stabiliser formalism. We study the spectral signatures of chaotic dynamics within fragments and the effect of walls on the spectral form factor in Section \ref{sec:spectral_chaos}. Lastly, we outline future directions involving generalisations to  models higher local Hilbert space dimension and higher dimensional lattices.

\section{Mathematical preliminaries \label{sec:math_prelims}}

%\chris{I'm rewriting this section.}
In this section, we review the mathematical formalism of Pauli and Clifford groups with an emphasis on connections to symplectic vector spaces and the two-qubit Clifford equivalence classes \cite{Tolar2018, Grier2022}.
We also introduce the notation used in the rest of the paper for specifying gates.

Consider the $1$-qubit Pauli group $\mathcal{P}_1 = \langle I,X,Y,Z \rangle$ where $I,X,Y,Z$ are the standard Pauli matrices.
The centre of this group $K(\mathcal{P}_1) = \{1, -1, i, -i\}$ is also its commutator subgroup and is also isomorphic to the finite field $\mathrm{GF}(4)$.
%\chris{Is this like a field?}
% where $\lambda \in \{1, -1, i, -i\}$ is a unit magnitude phase.
The $n$-qubit Pauli group, denoted by $\mathcal{P}_n$ is the set of operators that can be written as $\lambda \sigma_1 \otimes ... \otimes \sigma_n$, where $\sigma_i \in \{I, X, Y, Z\}$ and $\lambda \in K(\mathcal{P}_1)$.
%\chris{Is this a tensor product wrt Z(G)?}
In what follows, we shall treat elements of the Pauli group equivalent if they differ only by a unit magnitude phase factor which is tantamount to forming the quotient group $\mathcal{\Bar{P}}_n = \mathcal{P}_n / K(\mathcal{P}_n)$. The elements of $\mathcal{\Bar{P}}_n$ form a complete orthogonal basis for all $n$-qubit operators under the Hilbert-Schmidt norm. Therefore, one can write any operator $O$ supported on $n$-qubits as:
\begin{equation}
    O = \sum_{P \in \mathcal{\Bar{P}}_n} c_{P} \space P,
    \label{eq:Pauli_expansion}
\end{equation}
such that the coefficients satisfy $c_P = 2^{-n}\mathbf{Tr}[O^{\dagger} P]$ \cite{Siewert2022}. Throughout this work, we refer to the \textit{support} of an operator as the set of qubits where it acts non-trivially.

The $n$-qubit Clifford group, $\mathcal{C}_n$, is the subset of unitary operators that map Pauli group elements among one another under adjoint action. More formally, $C \in \mathcal{C}_n$ if $C P C^{\dagger} \in \mathcal P_n$ for all $P \in \mathcal{P}_n$. Due to the Pauli commutation relations, $n$-qubit Pauli operators are also elements of $\mathcal{C}_n$ that fulfil: $QPQ^{\dagger} = \lambda P$ for $Q, P \in \mathcal{P}_n$ and $\lambda \in K(\mathcal{P}_n)$. Most of our analytical results will be concerned with operator spreading, that is, the change of operator support under circuit dynamics. For this, it is useful to form the quotient group of Clifford with Paulis, $\mathcal{\Bar{C}}_n = \mathcal{C}_n / \mathcal{P}_n$ in order to neglect unit magnitude phases. 

The one-qubit Pauli group can be generated by just two traceless Paulis under multiplication, say, $X$ and $Z$, such that one can write $P = X^{p} Z^{q}$ where $p, q \in \{0, 1\}$. This suggests that one can represent any $P \in \mathcal{\Bar{P}}_n$ by a binary vector of size $2n$: $\mathbf{b} = (p_1, q_1, p_2, q_2, ..., p_n, q_n)$. More formally, the group isomorphism $\mathcal{\Bar{P}}_n \cong \mathds{Z}_2^{2n}$ holds. In this discrete phase-space picture, the Cliffords are represented by binary symplectic matrices of size $2n$, $\text{Sp}(2n)$, under the group isomorphism $\mathcal{C}_n / \mathcal{P}_n \cong \text{Sp}(2n)$. The constitutive relation for symplectics is the following: $S \in \text{Sp}({2n})$ if $S J S^T = J$ where $J$ is the binary symplectic form defined as $J = \bigoplus_{i=1}^n X$. This means that the order of $\text{Sp}({2n})$ is linear in $n$ for an $n$-qubit system which is the foundation of the classical simulability of Clifford circuits in polynomial time under the Gottesman-Knill theorem \cite{Aaronson2004, Tolar2018}.

In particular, we will analyse circuits built out of the two-qubit Clifford group $\mathcal{C}_2$. From the point of view of operator spreading, we will distinguish the separable subgroup $\mathcal{C}_1 \times \mathcal{C}_1$ (which preserves the support of operators) and the rest of the group. With this in mind, we form equivalence classes of gates with the elements of $\mathcal{C}_1 \times \mathcal{C}_1$. For convenience, we will take the representative elements of the classes as the identity, Controlled-$Z$ gate ($\CZ$), the $\SWAP$ gate and $\FSWAP = \CZ \circ \SWAP $. We define these gates on how they act on the generators of the two-qubit Pauli group $\mathcal{\Bar{P}}_2$ in Table \ref{tab:clifford_classes}. Note that this assignment is not unique as one can pick other members of the class as the representative element (e.g. Controlled-$X$ instead of Controlled-$Z$ for the second class). We also show them as tensor-diagrams on Figure \ref{fig:clifford_equiv_classes}., including the single-qubit Clifford degrees of freedom. This picture will be beneficial in constructing the localisation conditions in the following section.

\begin{table}[tbp]
    \centering
    % \begin{tabular}{c|c|c|c|c|c}
    %      Gate $C$ & $\mathds{P}(\mathrm{Class}(C))$ & $C  ( Z \otimes \mathds{1}) C^{\dagger}$ & $C (X \otimes \mathds{1}) C^{\dagger}$ & $C (\mathds{1} \otimes Z) C^{\dagger}$ & $C (\mathds{1} \otimes X) C^{\dagger}$ \\ \hline
    %      $\mathds{1}$ & 1/20 & $Z \otimes \mathds{1}$ & $X \otimes \mathds{1}$ & $\mathds{1} \otimes Z$ & $\mathds{1} \otimes X$ \\
    %      $\CZ$ & 9/20 & $Z \otimes \mathds{1}$ & $X \otimes Z$ & $\mathds{1} \otimes Z$ & $Z\otimes X$ \\ 
    %      $\SWAP$ & $1/20$ &$\mathds{1} \otimes Z$ & $\mathds{1} \otimes X$ & $Z \otimes \mathds{1}$ & $X \otimes \mathds{1}$ \\ 
    %      $\FSWAP$ &$9/20$ & $\mathds{1} \otimes Z$ & $Z \otimes X$ & $Z \otimes \mathds{1}$ & $X \otimes Z$ \\ 
    % \end{tabular}
    \begin{tabular}{c|cccc}
      Gate $C$ & $\mathds{1}$ & $\CZ$ & $\SWAP$ & $\FSWAP$ \\
      \hline
      $C(Z {\otimes} \mathds{1})C^{\dagger}$ & $Z {\otimes} \mathds{1}$ & $Z {\otimes} \mathds{1}$ & $\mathds{1} {\otimes} Z$ & $\mathds{1} {\otimes} Z$ \\
      $C(X {\otimes} \mathds{1})C^{\dagger}$ & $X {\otimes} \mathds{1}$ & $X {\otimes} Z$ & $\mathds{1} {\otimes} X$ & $Z {\otimes} X$ \\
      $C(\mathds{1} {\otimes} Z)C^{\dagger}$ & $\mathds{1} {\otimes} Z$ & $\mathds{1} {\otimes} Z$ & $Z {\otimes} \mathds{1}$ & $Z {\otimes} \mathds{1} $ \\
      $C(\mathds{1} {\otimes} X)C^{\dagger}$ & $\mathds{1} {\otimes} X$ & $Z {\otimes} X$ & $X {\otimes} \mathds{1}$ & $X {\otimes} Z$ \\
      \hline
      $\mathds{P}(\mathrm{Class}(C))$ & 1/20 & 9/20 & 1/20 & 9/20 \\
    \end{tabular}
    \caption{\small Representation of two-qubit Clifford equivalence classes w.r.t product gates, ie. $C \leftrightarrow (a_1\otimes a_2)C(a_3 \otimes a_4)$ for $a_i \in \mathcal{C}_1$. $\mathds{P}$ denotes the probability of drawing an element of a class under uniform measure \cite{Grier2022}. }
    \label{tab:clifford_classes}
\end{table}

\begin{figure}
    \centering
    \includegraphics[width=\linewidth]{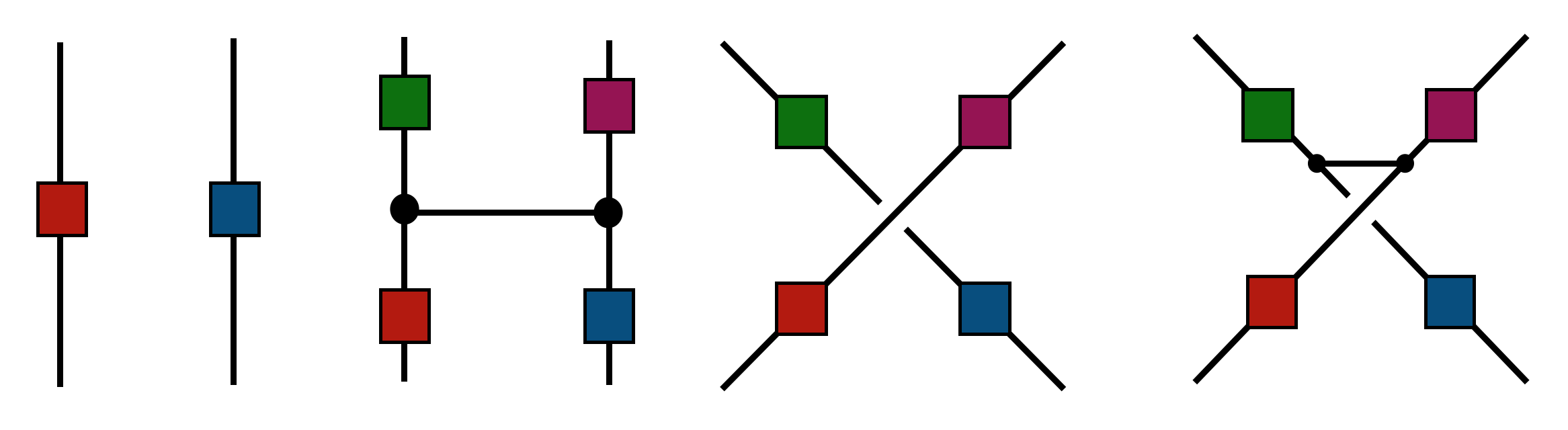}
    \caption{\small Equivalence classes of $\mathcal{C}_2$ with respect to product Cliffords as tensor diagrams. Identity-like, $\CZ$-like, $\SWAP$-like and $\FSWAP$-like gates shown respectively. Each leg represents a single qubit degree of freedom with multiple lines showing tensor product spaces. Boxes show elements of $\mathcal{C}_1$. Sampling the rectangular boxes under uniform measure on $\mathcal{C}_1$ respects the multiplicity of classes in Haar-sampling of ${\mathcal{C}}_2$ \cite{Grier2022, Mele2023}. }
    \label{fig:clifford_equiv_classes}
\end{figure}

\section{Model}
\label{sec:floquet_model}

We generalise the Floquet circuit model of \textit{Farshi et al.}~\cite{Farshi2022_1D, Farshi2022_2D} to include non-Clifford perturbations applied randomly on a one-dimensional qubit chain, as in Figure~\ref{fig:floquet_circuits}. Consider the evolution operator of a single Floquet period,
\begin{equation}
    U =
    \prod_{i=-\infty}^{+\infty} R_i
    \prod_{i\text{ odd}} C_{i, i+1}
    \prod_{i\text{ even}} C_{i, i+1}
    \text{,}
    \numberthis \label{eq:floquet_unitary}
\end{equation}
where we have used the following notation: $C_{i, i+1}$ denotes a Clifford gate which only acts on qubits $(i, i+1)$ non-trivially. Each of these is independently and identically distributed under a uniform measure on the set of \emph{non-product} Clifford gates.
% \begin{equation} % too much detail
%   \mathds{P}(C_{i, i+1}) = \begin{cases}
%     0 \text{, if } C_{i, i+1} \in \mathrm{Class}(\mathds{1}) \\
%   1/\left (|\mathcal{C}_2|-|\mathcal{C}_1|^2 \right) \text{ otherwise.}
% \end{cases}
% \label{eq:clifford_sampling_measure}
% \end{equation}
We exclude the product Clifford gates as they cause the circuit to decouple into a spatial product thereby preventing operator spreading in a trivial manner.
The single-qubit unitary $R_i$ acting on the $i$-th qubit is sampled from the following distribution,
\begin{equation}
  R_i = \begin{cases}
      \mathds{1} & \text{ with probability } 1-p \\
      % R({\epsilon_i, \mathbf{c}_i}) & \text{ with probability } p, \\
      % \exp \left( -i\epsilon_i \mathbf{c_i} \cdot \Vec{\sigma} / 2 \right) & \text{ with probability } p,
      H_i & \text{ with probability } p,
  \end{cases}
  \label{eq:single_qubit_haar}
\end{equation}
where $H_i$ are i.i.d. from the Haar distribution uniform over $\mathrm{U}(2)$. The evolution operator up to time $t$ is given by $U^t$ by supposing Floquet dynamics.

\begin{comment}
    The numerical results presented later in the paper were produced under a closely related measure where a uniformly sampled normalised rotation axis from $\mathds{R}^3$ and a uniformly sampled angle on $[0, 2\pi)$ is used as parameters of a Bloch-sphere rotation. Although this choice allows an additional control over the perturbation strength over the range of $\epsilon$, we haven't made use of this fact. This distribution is somewhat biased towards the identity compared to the Haar distribution, which we believe would be a more transparent choice for understanding perturbation-induced delocalization in the model. Our figures will be updated accordingly in due course although we expect they won't change our conclusions in any significant way. 
\end{comment}
% where $\Vec{\sigma} = (X, Y, Z)$ is the vector of traceless Pauli matrices,  $\epsilon_i$ is drawn uniformly from $[0, 2\pi)$ and $\mathbf{c}_i$ is a uniformly drawn unit vector in $\mathbb{R}^3$.
%\chris{The current versions of the figures however show data calculated with a slightly biased distribution and will be updated.}
The perturbation gates $R_i$ are potentially non-Clifford unitaries which will be used to break up the particular structure of Clifford circuits to test the robustness of the phenomena we observe. $p$ will serve as the control parameter in studying the localised/delocalised phase of the model. We choose the Haar measure for perturbations to understand the phases in the model for generic perturbations although we will remark later that other, biased, measures can be used to enhance the localisation in the model.

% \section{Wall configurations}
\section{Localisation as fragmentation}
\label{sec:wall_configs}

\begin{figure}
    \centering\includegraphics[width=0.75\linewidth]{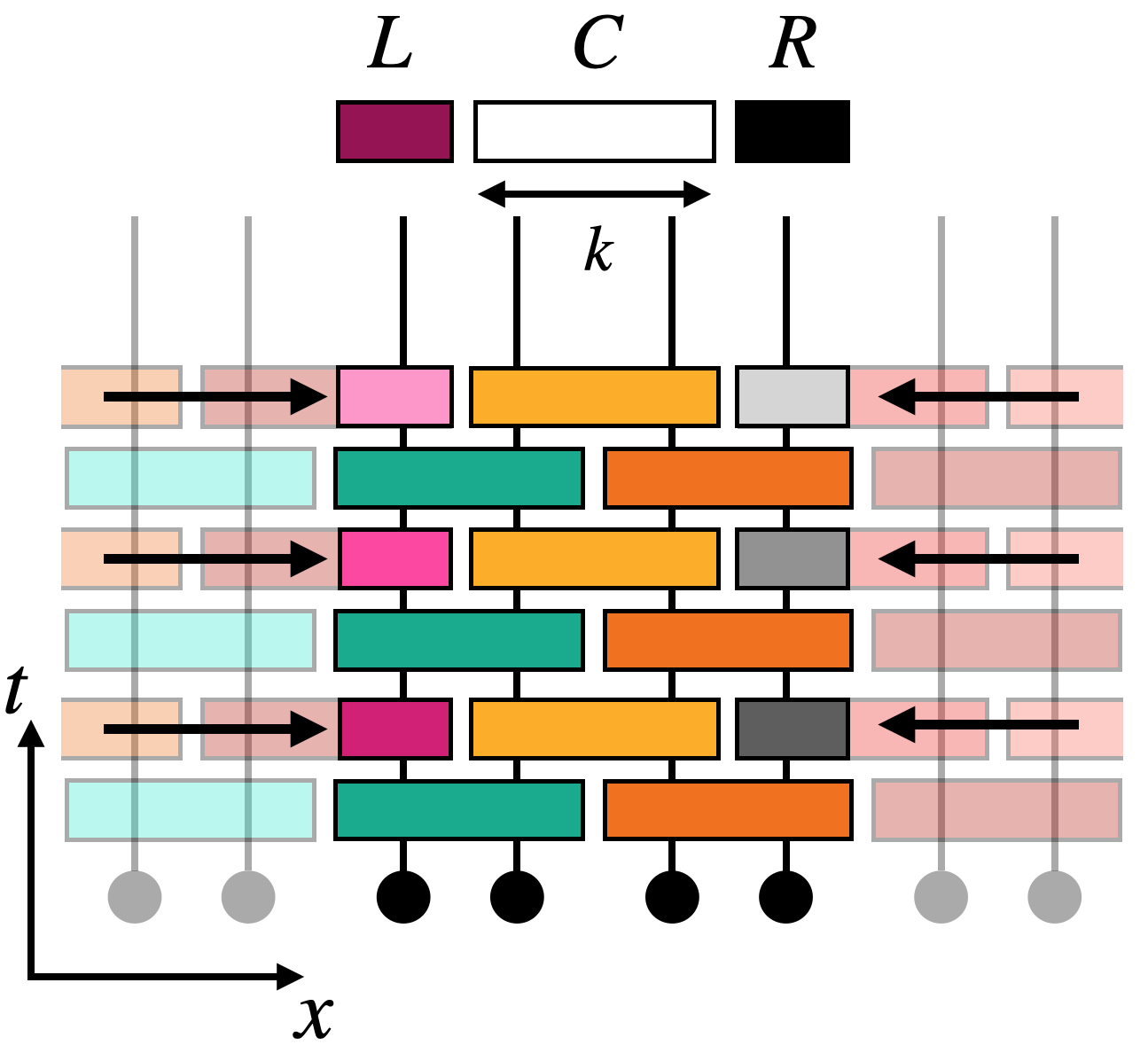}
    \caption{\small Representation of $k$-walls. We model the circuit environment as the injection of operators into the left ($L$) and right ($R$) qubit subspaces of the wall. Localisation in the circuit can be understood by constraining the inner gates to stop the spreading of arbitrary injected operators at the arrow locations. This ensures that localising gates can be stably embedded into larger circuits. }
    \label{fig:circuit_signals}
\end{figure}

Our paper concerns a particular form of dynamics in periodic systems where the operator spreading becomes exactly arrested, failing to percolate across the system. In this section, we formalise this idea with the concept of a $k$-wall as a unitary which stops the spreading of an arbitrary operators. We study the general properties of these walls in the context of Clifford circuits with special focus on the invariant subspaces of wall unitaries. This will serve two purposes. First, we show that the majority of wall gates harbour local conserved charges in the brickwork circuit model. Second, that in a random circuit ensemble, the space of Pauli operators is fragmented in the unitary dynamics where the invariant subspaces are enclosed by wall configurations. We study the structure and  sampling probability of wall configurations for two-layer Floquet circuits in \Cref{ssec:shallow_one_walls}--\ref{ssec:shallow_k_walls}. Due to the one-dimensional geometry such fragments are spatially localised qubit subspaces which are stable against random perturbations as shown in \Cref{ssec:stability}. In \Cref{ssec:localisation_length}, we calculate the typical size of fragments which is equivalent to localisation length of randomly placed local operators in the circuit.

\begin{comment}
This section considers the general conditions under which an operator with non-local support localises under the Floquet circuit evolution in the unperturbed $p=0$ model. As any operator decomposes as a superposition of Pauli operators, we will focus on the localisation of elements of $\mathcal{\Bar{P}}_n$.
This naturally suggests the definition of a width $k$ wall, i.e.~a gate configuration which generates localisation for an arbitrary operator.
\end{comment}

% Definition here needs to be general because this is our motivating object of study
\begin{definition}[$k$-walls]
  \label{def:walls}
  Let $C$ be a subsystem consisting of $k$ sites which splits the system into three pieces: $L$ for the left environment, $C$ for the wall subsystem and $R$ for the right environment which are all nontrivial, proper and connected subsystems.
  A circuit $U$ has a left wall of width $k$ (or $k$-wall) around $C$ if it arrests the spreading of operators in the following sense:
  for all operators $A$ on $L$ and times $t>0$ there exist operators $A'(t)$ on $LC$ such that,
  \begin{equation}
    U^t \left( A_L \otimes \mathds{1}_{C} \otimes \mathds{1}_{R} \right) U^{-t} = A'(t)_{LC} \otimes \mathds{1}_{R}
    \text{.}
  \end{equation}
  Additionally, we say that a wall is irreducible if it doesn't contain any subinterval which is itself a wall.
  Analogously we also define right $k$-walls by taking initial operators in $R$ and demanding they never spread to $L$.
\end{definition}

\begin{comment}
From the definition, one can trivially extend the localisation property to an arbitrary operator left to the wall due to locality. This can also be seen from the fact that an arbitrary number of random gates can form a reducible $k$-wall as long as the last two are a $1$-wall. For this reason, we will exclusively focus on irreducible wall configurations that we use interchangeably with the term wall. Naturally, $k=0$ walls would be equivalent to product gates that stop the spreading of operators immediately upon encountering the wall. In the following, we derive the constraints on the constituent gates and calculate the probability of encountering wall configurations. 

Then, we discuss how localising gates lead to quasi-local conservation laws which, by virtue of the Clifford circuit dynamics, can be done in a computationally inexpensive way.
Finally, we study the stability of the localisation mechanism in subsection \ref{ssec:stability} by developing the argument for a fragmented Pauli space due to the circuit disorder.
\end{comment}

Let us start by discussing the trivial case of $0$-walls which decouple the qubit chain to a tensor product of unitaries.

\begin{lemma}[$0$-walls]
  \label{lemma:product_walls}
  The $0$-walls are exactly the product unitaries.
\end{lemma}
\begin{proof}
  Despite seeming physically obvious the proof is both quite involved and unrelated to what follows so we elucidate these instances in \Cref{sec:0-walls}.
\end{proof}

% Motivate moving from general case to Clifford specificity
We remark that it is reasonable to suppose that this form of strict localisation in operator spreading is unusual, with the typical behaviour being ballistic. This is due to the fact that the majority of non-product Clifford gates (the $\mathrm{SWAP}$, $\mathrm{FSWAP}$ classes) are dual-unitary which have been discussed extensively in the literature in the context of maximally chaotic models \cite{Prozen2020, Bertini2018, Bertini2019, Bertini2021}.
In what follows, we will focus on the emergence of non-trivial walls in Clifford circuits and their destabilisation with generic single-qubit perturbations. 
These are not the only cases in which walls can appear, and these perturbations can clearly create walls albeit with probability zero.
However we will discount the possibility of walls occurring from some other mechanism under the assumption these are negligible in our particular circuit distribution.

% Motivate the signal definition and how we can study small circuits to understand large ones.

The rationale behind studying walls is the idea that understanding configurations of a few gates that stop an arbitrary operator from spreading will enable us to embed these into an arbitrary size circuit segment which still generate localisation. This will lead to a non-trivial decoupling of a qubit chain with an associated signature in entanglement spreading and spectral statistics.

Consider cutting out a central circuit segment and an adjacent buffer region around it, as on Figure \ref{fig:circuit_signals}. To faithfully represent the dynamics within this restricted region, one would need to model the operator dynamics at the wall edges as a stochastic process which generates a time-dependent input operators from the environment within the central subspace in the spirit of the Mori-Zwanzig formalism \cite{Zwanzig2001}. We erase knowledge of the environment circuits, modelling operator flow and back-flow as an arbitrary signal injecting operators into one of the buffer regions and considering the time evolution of Pauli subgroups at the injection points. A wall that is stable against all signals will be stable in any embedding. 

In the following series of lemmas, we rigorously develop this idea by considering the left and right subspaces corresponding to the evolution of initially left/right operators. This will allow us to establish that Clifford walls stop operator spreading from either side (\Cref{lemma:twosided}). This will be instrumental in showing the existence of local conserved charges in the circuit hosted in the central subspace (\Cref{lemma:intersection}). These are single-qubit Pauli matrices for the case of $1$-walls (\Cref{lemma:conserved_pauli}), which are most likely to appear in disorder realisations in uniformly sampling the non-product Clifford operators. We will also construct local conservation laws for higher-order walls.

We will provide two definitions of walls for Clifford circuits -- one is the \emph{initial-condition} characterisation analogous to the general one given in \Cref{def:walls} and the other is the \emph{signal} characterisation.
Let $W$ denote the wall subcircuit's Clifford operator in the associated symplectic vector space. In the \emph{initial condition characterisation},
there is a subspace $G_\text{left}$ of the symplectic vector space associated to $C$ such that, for each pair $t$ and $u$, there are $v$ and $g\in G_\text{left}$ such that
\begin{equation}
    W^t \left(u \oplus 0 \oplus 0\right) = v \oplus g \oplus 0
    \text{.}
\end{equation}
We refer to $G_\text{left}$ as the (left) internal subspace and for the right wall condition we denote the corresponding internal subspace as $G_\text{right}$.

  In the \emph{signal characterisation} we allow for arbitrary operators to be inserted throughout time on either the $L$ or $R$ subsystems for left and right walls respectively. This means that all signals or functions $u(t)$, there exists $v$ and $g$ in some $G_\text{left}$, which is a subspace of the symplectic vector space associated to $C$, such that
  \begin{equation}
    \sum_t W^t (u(t) \oplus 0_C \oplus 0_R) = v \oplus g \oplus 0_R
    \text{.}
  \end{equation}

\noindent We now show that these two approaches are equivalent.
\begin{comment}
  First we will give explicit statements of these characterisations.
  The \emph{initial-condition} characterisation of a (left) wall is as follows: there exists $v$ and $g \in G_\text{left}$ which is a subspace of the symplectic vector space associated to $C$ for all $t$ and $u$ such that,
  \begin{equation}
    W^t \left(u \oplus 0 \oplus 0\right) = v \oplus g \oplus 0
    \text{.}
  \end{equation}
  We refer to $G_\text{left}$ as the (left) internal subspace and for the right wall condition we denote the corresponding internal subspace as $G_\text{right}$.

  The \emph{signal} characterisation of a (left) wall in a symplectic matrix $W$ is that for all signals or functions $u(t)$, there exists $v$ and $g$ in some $G_\text{left}$ which is a subspace of the symplectic vector space associated to $C$, such that
  \begin{equation}
    \sum_t W^t (u(t) \oplus 0_C \oplus 0_R) = v \oplus g \oplus 0_R
    \text{.}
  \end{equation}
\end{comment}
\begin{lemma}
  A wall subcircuit is a wall to input signals if and only if it is a wall to initial conditions.
  \label{lemma:signals}
\end{lemma}
\begin{proof}
One can formulate the \emph{initial-condition} characterisation in terms of some arbitrary polynomials $p_i$ for a basis $e_i$ of $L$,
\begin{equation}
p_i(W) (e_i \oplus 0_C \oplus 0_R) = v \oplus g \oplus 0_R
\text{.}
\end{equation}
Because we can superpose signals, this motivates $G_\text{left}$ being closed as a vector space.

From the signal characterisation we choose monomials $p_i(x) = u_i x^t$ where $u_i$ are coefficients of $u$ in the basis $\{e_i\}$ to recover the initial-condition characterisation.
In the reverse direction, we interpret each time and initial condition as a collection of monomials and together these form a basis of the polynomials so we can take linear combinations of the equation to recover the condition for arbitrary polynomials.
\end{proof}

As a consequence, for our main circuit geometry of study defined by \Cref{eq:floquet_unitary}, all $k$ walls can be understood by looking at reduced circuits of $k+2$ sites enclosing the wall with a single site for each of $L$ and $R$.
Before starting to look at the details of this geometry we have a few further general results about Clifford walls.

\begin{lemma}
  All walls in local Clifford circuits are two-sided.
  \label{lemma:twosided}
\end{lemma}
\begin{proof}
  Assume that for the interval $C$ the left wall condition holds in some Clifford circuit modelled by symplectic matrix $W$.
  For any $t$ and any $l$, we have $(0_L \oplus 0_C \oplus r)^T J W^t (l \oplus 0_C \oplus 0_R) = 0$ because $l$ can't spread to $R$.
  Using the symplectic property of $W$, we can transfer the time evolution from $l$ to $r$, although time is reversed, i.e. $(l \oplus 0_C \oplus 0_R)^T J W^{-t} (r \oplus 0_C \oplus 0_R) = 0$.
  The locality of the circuit means that we can always cut out a finite circuit and equivalently study that instead.
  In that finite circuit the time evolution orbit of any symplectic vector is finite so from all negative times we can make all positive times as required.
  We can take $l$ from a complete basis for the left subspace, hence the time evolution of $r$ is $J$-orthogonal to the left subspace.
  Since the left subspace is symplectic, it is disjoint from its symplectic complement and therefore the time evolution of $r$ is always orthogonal to the left subspace.
  This is equivalent to the right wall condition and the reverse implication follows by symmetry.
\end{proof}

\begin{lemma}[Local conserved charges]
  A wall in a local Clifford circuit with a non-trivial intersection $G_\mathrm{left} \cap G_\mathrm{right}$ of internal subspaces has corresponding non-trivial local conserved quantities.
  \label{lemma:intersection}
\end{lemma}
\begin{proof}
  Choose some $g \in G_\text{left} \cap G_\text{right}$.
  From membership of these subspaces, there exists some left input signal and some right input signal that independently produce $g$ in the centre at some time, together with some additional operator content either to the left or the right respectively.
  Hence, an impulse can be added to each signal at that time to remove that additional operator content and create a pair of left- and right- signals which both produce $0 \oplus g \oplus 0$ after some time.
  %final part of each signal can be always chosen to make the projection onto the input subspace (i.e. left or right respectively) trivial.
  Subsequently, the signals are chosen to be zero, so the evolution is exactly the evolution starting from $0 \oplus g \oplus 0$.
  From the two different ways of starting this evolution we know that it must be both blocked from reaching the right and blocked from reaching the left -- so it must remain in the central subsystem $C$ at all times.
  By integrating the corresponding operator along the orbit together with some phase function we can produce an operator on which the evolution acts as a scalar.
  A projection onto this one-dimensional representation is then a conserved superoperator.

  Let us go through the construction of conserved quantity more explicitly.
  Let $M$ be an operator whose simplectic representation is $g$. This fixes $M$ up to a global phase. Also, we know that it is inside the wall, so $M\in \mathcal{P}_k$.
  Since $g$ remains within the wall for all times, so does $M(t) = U^t M U^{-t}$ where $U$ is the unitary defining the circuit.
  Due to finiteness, at some point the operator comes back to itself at some finite time $\tau > 0$ up to a phase $\mathrm{e}^{i\theta}$ for some $\theta \in \mathbb{R}$,
  \begin{equation}
    M(\tau) = \mathrm{e}^{i\theta} M(0)
    \text{.}
  \end{equation}
  As a result, one can generate a quantity $Q_M$ by integration along the orbit,
  \begin{equation}
    Q_M = \sum_{t=0}^{\tau-1} \mathrm{e}^{-i\theta t / \tau} M(t)
    \text{,}
  \end{equation}
  which satisfies $U^t Q_M U^{-t} = \mathrm{e}^{i \theta t / \tau}Q_M$ for all times $t$.
  This is a closed one-dimensional subspace of the operator algebra and the corresponding projection superoperator is a conserved charge of the operator dynamics.
  The operator $Q_M$ will either be conserved, or will undergo coherent phase oscillations.
  For our purposes, we will treat these coherently oscillating quantities as conserved charges with the view that overall phases do not matter.
  %\chris{This isn't quite what we were saying before.}
\end{proof}

\noindent We remark that conserved quantities constructed in this manner have also been recently discussed in the literature for Floquet dual-unitary circuits \cite{HoldenDye2023}.

\begin{lemma}[Internal orthogonality]
  For a wall in a local Clifford circuit, the two internal subspaces are $J$-orthogonal, i.e. $G_\text{left} J G_\text{right} = 0$.
  Equivalently, considered as Pauli subgroups, the two groups pairwise commute.
  \label{lemma:internal_orthogonality}
\end{lemma}
\begin{proof}
  Using the same techniques as in \Cref{lemma:twosided}, we can \emph{partially} transfer the time evolution to see that all vectors generated by left and right signals are $J$-orthogonal.
  Hence for any $l$, $r$, $g \in G_\text{left}$ and $g' \in G_\text{right}$ we have that
  $(l \oplus g \oplus 0_R) J (0_L \oplus g' \oplus r) = 0$.
  This equation breaks up into three terms for each of the symplectic subspaces of $L$, $C$ and $R$.
  The terms for $L$ and $R$ individually vanish because at least one of the vectors in the $J$-product is zero when projected into that subspace.
  Hence, the $C$ term must also individually vanish which establishes $J$-orthogonality for the internal subspaces.
\end{proof}

Finally we have some results specific to Clifford $1$-walls that are independent of the circuit geometry.

\begin{lemma}[Internal subspace for 1-walls]
  Let $W$ be an irreducible $1$-wall.
  The left and right internal subspaces of the wall are identical, have dimension 1 and therefore corresponds to a group generated by some single traceless Pauli matrix.
  \label{lemma:internal_pauli}
\end{lemma}
\begin{proof}
  First suppose that $G_\text{left}$ is the trivial subspace.
  In this case we could remove the central site and join it to the $R$ subsystem to create a $0$-wall, contradicting the wall being irreducible.
  Next suppose that $G_\text{left}$ is the full $\mathbb{Z}_2^2$.
  In this case we could clearly find some signal to generate any element of $G_\text{left}$ combined with any element of the left symplectic subspace.
  This would allow us to join the central site to the $L$ subsytem forming a $0$-wall and again contradicting that the wall be irreducible.
  The internal subspace must therefore be a proper and non-trivial subspace of symplectic subspace for $C$, these all correspond to the group generated by a single Pauli matrix.
  By symmetry the same is true for the right internal subspace, but it remains to show that the subspaces coincide.
  Using \Cref{lemma:internal_orthogonality}, we can see that $G_\text{left} = G_\text{right}$ because these groups are their own centralizers.
\end{proof}

Note that the first part of the proof generalises to $k\,{>}\,1$, where it shows that the internal subspaces are non-trivial proper subspaces for any irreducible $k$-wall.

\begin{corollary}
  Every $1$-wall conserves some Pauli matrix.
  \label{lemma:conserved_pauli}
\end{corollary}
\begin{proof}
  From the \Cref{lemma:internal_pauli}, the intersection of internal subspaces is generated by a single traceless Pauli operator.
  So we use \Cref{lemma:intersection}) to see that the integral along the orbit of this operator is a local conserved charge.
  The orbit must however remain inside the internal subspace but there are no other traceless elements so this Pauli operator must itself be the local conserved charge.
\end{proof}

\subsection{Structure of shallow 1-walls \label{ssec:shallow_one_walls}}

% Here's some old remarks that are geometry specific
\begin{comment}
By exploiting the locality of the circuit, the spreading of operator support is controlled by re-shaping the brickwork circuit elements as a \enquote*{staircase} while respecting the periodicity of the operator, as on Figure \ref{fig:staircase_circuit}.
Although the existence and sampling probability of $k=1$ walls was already established in \cite{Farshi2022_1D, Farshi2022_2D} by exhaustively enumerating symplectic matrix combinations, we now provide an explicit analytical construction for $k=1$ and $k=2$ and study the generalisation for higher $k$.
This will allow us to conclude the section by calculating the typical localisation length of operators.
\end{comment}

\begin{figure}
  \centering\includegraphics[width=\linewidth]{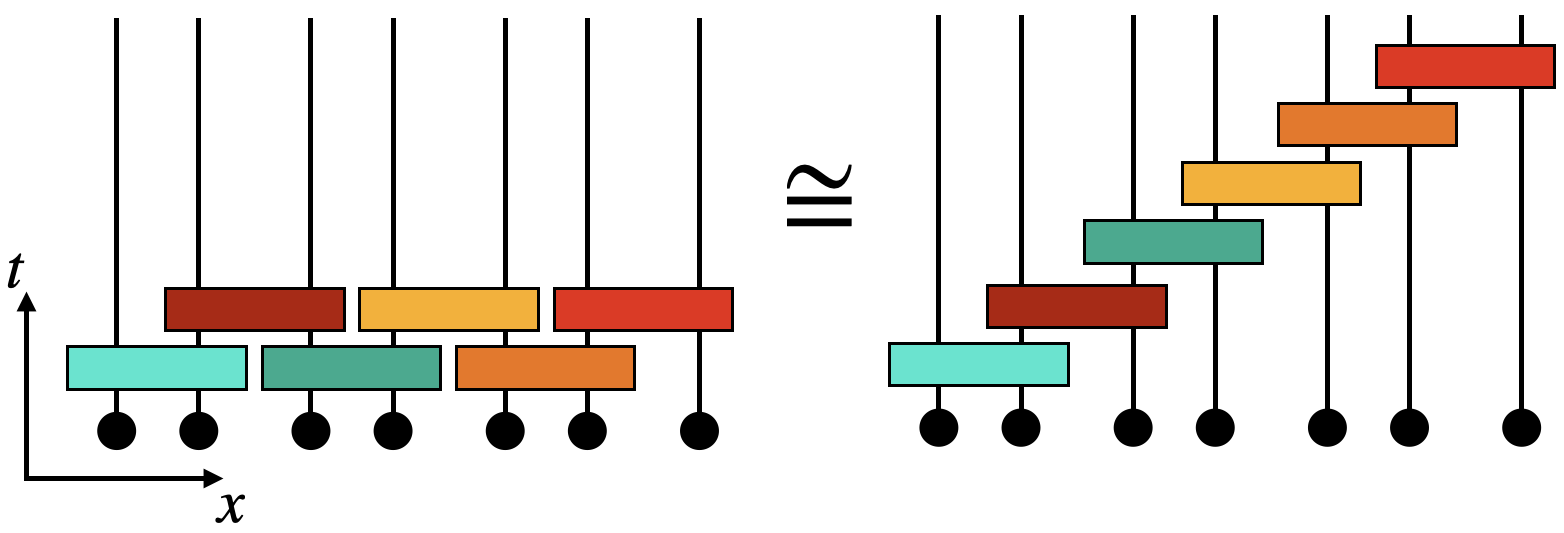}
  \caption{\small In the $p=0$ (Clifford) limit, the two Floquet operators shown in the figure give rise to the same evolution operator, except at finite time boundaries. The Floquet operator in \enquote{staircase} form (right-hand side) reflects the causal structure of operator spreading.
  %Reduced, \enquote*{staircase},  representation of the $p=0$ Clifford circuit that generates operator spreading. The time-periodic nature of the circuit allows us to reshape the initial layer into staircase form, which reflects the logical structure of the operator spreading, i.e.~feeding-forward operators between the gates.
  }
  \label{fig:staircase_circuit}
\end{figure}

We will now turn to Clifford circuits of the form given by  \Cref{eq:floquet_unitary}, where constituent gates are sampled uniformly but with separable Clifford gates excluded.  
In the following we will use \emph{shallow} to mean exactly as shallow as this Floquet circuit layer, rather than other notions such as more general finite depth circuits.
The term \emph{shallow} will also disallow the Clifford gates being separable consistent with the measure. We will also describe the general properties of $k>1$ walls which will allow us to bound their probability of formation and comment on the conserved charges they host.

A useful technique is to redefine the Floquet operator without changing the evolution operator. Specifically, note that the Floquet operator depicted in the left-hand side of \Cref {fig:staircase_circuit} gives rise to the same brickwork circuit than the Floquet operator depicted in the right-hand side of \Cref {fig:staircase_circuit}, except at the time boundaries. The staircase form of the Floquet operator is useful to analyse the spreading of local operators and to characterise the presence of $k$-walls.
%This can be done using Clifford gauge transformations to the $C$ subsystem, or visually by slicing the Floquet circuit into repeated layers along some other space-like curve -- clearly this leaves the wall conditions unchanged since it is a one-to one mapping.
Using this idea we can provide some constraints on the classes of gates which constitute an irreducible $k$-wall for these shallow circuits.

\begin{lemma}
  For any $k > 0$, a shallow Clifford irreducible $k$-wall begins and ends with $\CZ$-class gates.
\end{lemma}
\begin{proof}
  Consider the reduced circuit with a single qubit for each of $L$ and $R$.
  An initial operator on $L$ encounters the first gate and produces an operator on $L$ and the first site of $C$.
  We take this and project the corresponding symplectic vector to the subspace for that first site of $C$.
  Both of these steps are group homomorphisms so the image of $\mathds{Z}_2^{2n}$ through the composition is a subgroup of $\mathds{Z}_2^{2n}$.
  The image subgroup cannot be the trivial group because the gate is by assumption not separable.
  If the image group is the full group $\mathds{Z}_2^{2}$ then we can use a signal on $L$ to produce any signal output on the first site of $C$ to produce a $(k-1)$-wall, which contradict irreducibility.
  Hence, for an irreducible wall the image group is the proper subgroup generated by a single Pauli.
  This can only be achieved by a $\CZ$-class gate.
  The same argument applies on the right boundary to the wall.
\end{proof}

\begin{figure}
  \centering
  \includegraphics[width=0.48\textwidth]{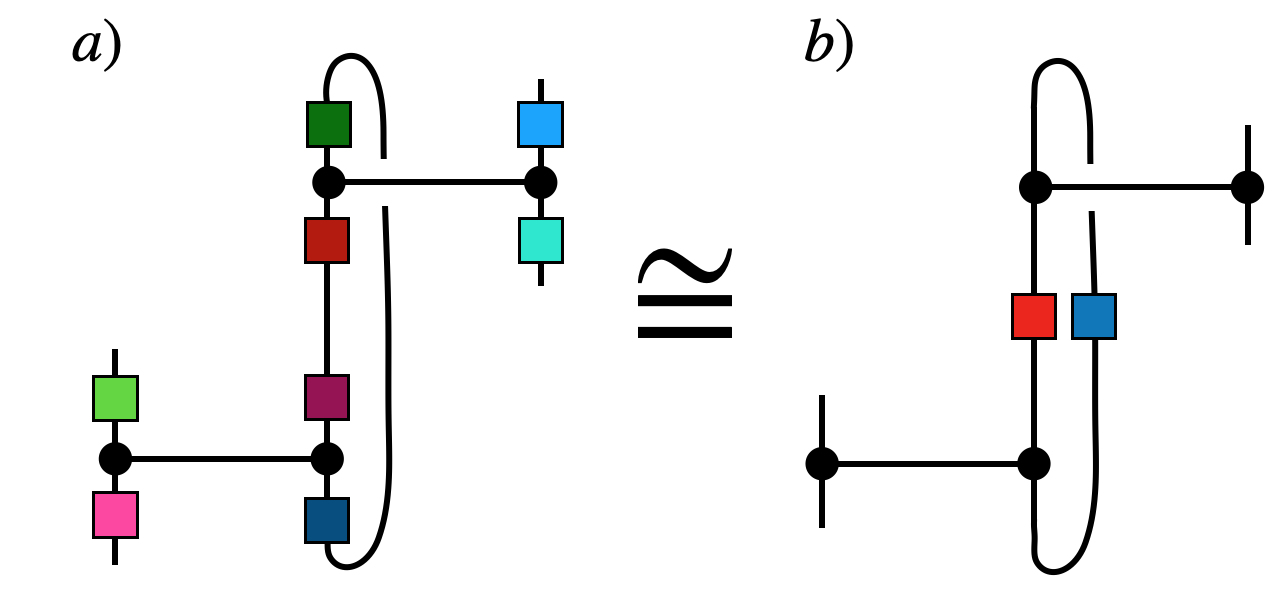}
  \caption{\small General form of $1$-walls by attaching two members of the $\CZ$-class. The loop on the central qubit captures the periodic nature of the circuit. Two compatible $\CZ$-like gates form a one-wall if their single-qubit Clifford degrees of freedom preserve the conserved Pauli on the central qubit as per \Cref{lemma:conserved_pauli}. Sampling under uniform $\mathcal{C}_1$ measure gives the same probability of sampling the configuration on  \textit{a)} and \textit{b)}, by Haar-invariance. }
  \label{fig:1wall_structure}
\end{figure}

As a result, all $1$-walls are made out of $\CZ$-class gates, shown on \Cref{fig:1wall_structure}.
The looping back connection simply represents the feeding of operators into the next layer of the circuit and shouldn't always be interpretted as a trace.
The simplest example is two $\CZ$ gates joined together as a staircase.
Since both the $\CZ$ gates leave the $Z$ operator fixed, as per \Cref{tab:clifford_classes}, this is the local conserved charge.

We can now formulate the structure of all shallow $1$-walls, using \Cref{lemma:conserved_pauli} together with the idea of circuit reshaping.
Essentially, two $\CZ$-like gates will make a wall if their single-qubit Clifford degrees of freedom are \enquote*{compatible}, i.e.~they respect the conservation of the central qubit's internal subspace.
Using this we may calculate the probability of forming a wall under uniform measure on non-product Cliffords by counting equivalences of $\CZ$-like gates.
To start, we can take the gate pair and decompose it into $\CZ$ gates together with a number of $\mathcal{C}_1$ unitaries, which we can combine (and amputate on the outside as they are irrelevant) to simplify the circuit.
% Consider the local evolution of the conserved Pauli on the intermediate wires on Figure \ref{fig:1wall_structure}. The two wires will respect the conservation of the subspace if they map under $\sigma_c \rightarrow \sigma \rightarrow \sigma_c$ for some $\sigma \in \{ X, Y, Z\}$. This happens with $1/3$
% probability in the single-qubit Clifford group giving a probability of $(1/3)^2 = 1/9$ of obtaining two compatible gates. Using the group axioms on Figure \ref{fig:1wall_structure} \textit{b)
% }, one can also see this by noting that $1/3$ fraction of $\mathcal{C}_1$ leaves $Z$ invariant up to phases.
Reshaping means that even on all the intermediate wires between Floquet layers for the central site there is a conserved charge.
We know that this conserved operator must always be proportional to $Z$ because otherwise they would escape through the $\CZ$ gates.
The remaining $\mathcal{C}_1$ fulfils this constraint with probability $1/3$, because that is the fraction of $\mathcal{C}_1$ which preserves $Z$ up to phases.
% We can now formulate the structure of \textit{all} $1$-walls, shown on \Cref{fig:one-walls}.
% Consider the local evolution of the conserved Pauli on the intermediate wires on Figure \ref{fig:1wall_structure}. The two wires will respect the conservation of the subspace if they map under $\sigma_c \rightarrow \sigma \rightarrow \sigma_c$ for some $\sigma \in \{ X, Y, Z\}$. This happens with $1/3$
% probability in the single-qubit Clifford group giving a probability of $(1/3)^2 = 1/9$ of obtaining two compatible gates. Using the group axioms on Figure \ref{fig:1wall_structure} \textit{b)
% }, one can also see this by noting that $1/3$ fraction of $\mathcal{C}_1$ leaves $Z$ invariant up to phases.
Sampling two $\CZ$-class members (without product gates) has probability $(9/19)^2$ yielding as final probability for sampling a $1$-wall,
\begin{equation}
  \mathds{P}(\text{$1$-wall}) = \frac{1}{9}\left( \frac{9}{19} \right)^2 \approx 0.025. 
\end{equation}
This is in exact agreement with the previously reported wall probabilities in \cite{Farshi2022_1D, Farshi2022_2D} (which were performed by counting wall instances in the symplectic group representation of two-qubit Clifford pairs), accounting for the removal of product gates. This corroborates that we have enumerated all wall instances for $k\,{=}\,1$.
From the tensor diagram picture, we can also see that these walls are \textit{two-sided}, see \Cref{lemma:twosided}, ie. they block the spreading of right localised operators as well as left ones since mirroring Figure \ref{fig:1wall_structure} horizontally respects the topology of the tensor.
% In fact, the loop of the tensor represents right-to-left spreading Pauli operators thus ensuring that only $\sigma_c$ propagates on that leg is equivalent to constraining the intermediate single-qubit Cliffords as above. 

\subsection{Structure of shallow 2-walls \label{sscec:shallow_two_walls}}

We are also able to explicitly construct the form of all $2$-walls and count their sampling probability. Naturally, any gate in the staircase form will make a $2$-wall if it also has a $1$-wall embedded into it. Here, we only consider walls which are not reducible to $1$-walls in this sense.
From the arguments of the previous section, these configurations necessarily have two $\CZ$-like gates at their left and right ends.

\begin{lemma}
  All the gates interior to a shallow Clifford irreducible wall are $\SWAP$ and $\FSWAP$ class gates.
\end{lemma}
\begin{proof}
  Suppose there is some $\CZ$-class gate part way through the staircase of such an ireducible wall.
  With a left initial condition we can cause the interior $\CZ$-class gate to emit an operator to the right, otherwise the wall is clearly reducible -- we call this activating the gate.

  % If we take the initial condition and integrate it through time it will activate the interior $\CZ$-class gate at some sequence of times.
  % We create a signal using delayed copies of the initial condition such that each subsequent activation is cancelled by the first activation of a delayed copy
  % We can remove all the activations subsequent to the first by creating a signal where we cancel each
  % If we take this as a signal and convolve it through time in some way, we can remove all subsequent activations of the gate, including any \emph{back-flow} from the right side of the gate into the left side and then back to the right side.
  % Take an arbitrary signal indended to the inserted onto the left side of the interior $\CZ$-gate.

  Now we focus on the behaviour of the right half of the circuit from the interior $\CZ$-gate to the right end.
  By superposing left initial conditions which activated the gate first at some time, we can produce any sequence of activation times of the $\CZ$-gate.
  If no non-trivial operator reaches $R$ for any sequence of activation times then clearly no signal directly placed on the left-side of the interior gate can spread to $R$ either.
  This means that there is an interior wall from the interior $\CZ$ to the right boundary, contradicting the assumption of irreducibility.
  Therefore there can be no interior $\CZ$-class gate.
  % Now we can imagine cutting out the subcircuit from the interior $\CZ$-class gate all the way to the right-side terminating $\CZ$-class gate of the wall.
  % The signal we have constructed is equivalent to a signal on the reduced system using some Pauli to activate the first gate and then fixing the input qubit to carry identity at all times.
  % We can also consider a signal on the reduced system that inserts the traceless Pauli which does not activate the first gate and then fixes the input qubit to carry identity.
  % Taking these two, together with all their time translations, and forming linear combinations to create all signals on the input qubit, we see that the subcircuit is itself a wall, which stretches from the right side of the interior $\CZ$-class gate to the left side of the terminating one.
  % This contradicts irreducibility.
\end{proof}

\begin{figure*}
    \centering
    \includegraphics[width=\textwidth]{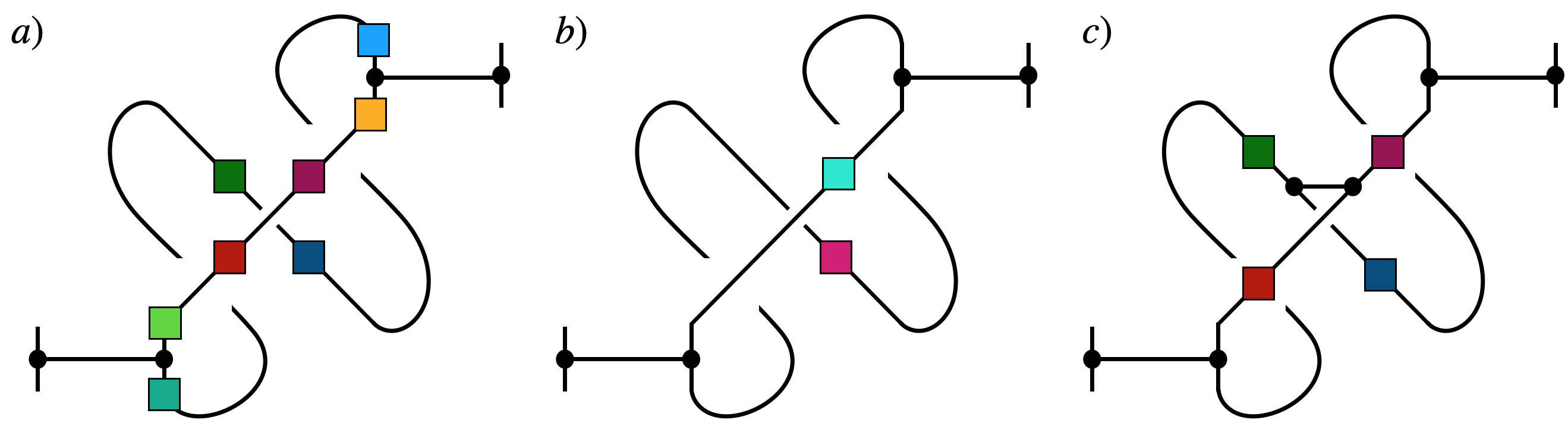}
    \caption{\small Structure of $2$-walls. Using the Haar-invariance of sampling $\mathcal{C}_1$ elements, we collapse the single-qubit Cliffords on \textit{a)} while respecting the topology of the tensor. This reduces $\SWAP$-like $2$-walls to the $\CZ$-like case on 
    \textit{b)}. For $\FSWAP$-like walls on \textit{c)}, we have manually constructed the consistency conditions in Appendix \ref{app:fswap_consistency} to find the probability of wall formation. }
    \label{fig:2wall_structure}
\end{figure*}

Therefore, the structure of all shallow 2-walls will have either a $\SWAP$ or $\FSWAP$ class gate in the middle, as shown in \Cref{fig:2wall_structure}.
We treat these cases separately. The challenge of constructing the wall instances is assigning the single-qubit Clifford degrees of freedom consistently, ie. to preserve that only a single traceless Pauli may arrive on the output $\CZ$-gate in order to preserve localisation. 

For the $\SWAP$ case, the intermediate wires decouple therefore the probability of obtaining a wall for uniformly sampled assignments of the single-qubit Cliffords is the same as for the $1$-walls. By inserting the probability of sampling a $\SWAP$-like gate under the model's Clifford measure, we get:
\begin{equation}
    \mathds{P}(\text{SWAP-like 2-wall}) = \frac{1}{9} \left(\frac{1}{19}\right)\left( \frac{9}{19} \right)^2.   
\end{equation}
These configurations therefore lack any \enquote*{interference} between Paulis on wires corresponding to forward and backward propagation. One may also count equivalent gate configurations by assigning a single-qubit Pauli subgroup $\{ \mathds{1}, \sigma\}$ to each wire independently to preserve localisation. This picture does not always hold for $\FSWAP$s.

The constructions we have found so far share a simplifying idea which we call \emph{interference-freedom}.
This is when the internal subspaces splits as a product over the interior qubits of $C$ so that the conserved Pauli subgroup is a Cartesian product of groups generated by a single traceless Pauli for each qubit.
In a interference-free wall, we can break apart the operator at any point in its evolution into individual Pauli operators and cast each forward in time independently.
None of these will escape the wall.
Furthermore, the looped diagrams we write in this case \emph{can} be interpreted as a trace, where the lines carry a vector space of dimension three corresponding to the three different Pauli subgroups they can be assigned to.
This idea allows one to calculate all non-interfering walls for any $k$ by a one-dimensional transfer matrix method.
Unfortunately, this does not generalise to all shallow walls.

For an \emph{interfering} wall, naturally, this does not work and the wall can only be seen by looking at the Pauli operators across $C$ in combination.
Viewing the symplectic vector space as a $\mathbb{Z}_2$-valued `wavefunction' over a discrete phase-space, the operator is a superposition of basis vectors and it is only the interference between these that ensures that they do not escape to the wrong side of the wall.
This motivates our terminology while the process is analogous to the destructive interference of real-space single-particle wavefunctions in Anderson localisation.
For a particular example, the 2-walls with a $\FSWAP$-like interior can exhibit interference.

In Appendix \ref{app:fswap_consistency}, we have enumerated all the consistent assignments for the single-qubit degree freedoms of $\FSWAP$-like gates.
Remarkably, this results in the same probability of obtaining compatible gates as before,
\begin{equation}
    \mathds{P}(\text{FSWAP-like 2-wall}) = \frac{1}{9^2} \left(5+4\right)\left(\frac{9}{19}\right)\left( \frac{9}{19} \right)^2
    \text{.}
    \label{eq:sampling_prob_fswap}
\end{equation}
The first term comes from the $5$ interference-free diagrams whereas the remaining $4$ contribution comes from circuits exhibiting interference.
% different wires may carry two-different Paulis which then destructively interfere by the time the operator arrives to the $\CZ$-gate of the wall. We show examples of this on Figure \ref{fig:fswap_interference}.
We have numerically verified that the above conditions are necessary and sufficient for constructing $2$-walls which yield the total probability to be of similar order to that of $1$-walls,
\begin{equation}
    \mathds{P}(\text{2-wall}) = \frac{1}{9} \left(\frac{10}{19}\right)\left( \frac{9}{19} \right)^2 \approx 0.013
    \text{.}
\end{equation}

%\chris{Integrate this into the 2 wall section.}
One can also construct at most $2$-local conservation laws for $k=2$. For the $\SWAP$-like case, this gives:
\begin{equation}
    Q_M^{\SWAP} = \mathds{1} \otimes \left ( \sigma_c \otimes  \mathds{1}  +  \mathds{1} \otimes \sigma_c \right )\otimes \mathds{1}, 
\end{equation}
where we have used the diagrammatic construction of Figure \ref{fig:2wall_structure} and note that the $\SWAP$-gates preserved the locality of the conserved quantity. For the non-interfering cases with $\FSWAP$, an analogous construction works as well. For the interfering $\FSWAP$-like cases, one can't construct non-local conservation laws as these instances don't have non-trivial intersections under \Cref{lemma:intersection}. 

\subsection{Structure of shallow \texorpdfstring{$k$}{k}-walls \label{ssec:shallow_k_walls}}

Using the previous constructions, one can generate arbitrary long walls by increasing the number of $\SWAP$/$\FSWAP$-like gates between the wall edges (that are not reducible to smaller ones). Accounting for the interference permitted by the $\FSWAP$-like gates, however, becomes an increasingly difficult task as $k$ increases which prevented us for calculating the exact probability of forming $k$-walls. We note that for any $k$-wall made up from $\SWAP$-like gates only, Haar-invariance gives the consistency probability to be $1/9$, as for $k=2$.

Counting only \emph{interference-free} walls gives the following lower bound for any $k$,
\begin{equation}
  \left(\frac{9}{19}\right)^2 \frac{1}{9}\left(\frac{6}{19}\right)^{k-1}
  \le \mathds{P}(\text{$k$-wall})
  % \le \left(\frac{9}{19}\right)^2 \left  (\frac{10}{19}\right)^{k-1}
  \text{,}
  % \nonumber
\end{equation}
and the following upper bound can be found using only the probability of sampling the correct sequence of equivalence classes required to form an irreducible wall of this length,
\begin{equation}
  \mathds{P}(\text{$k$-wall}) \leq \left(\frac{9}{19}\right)^2 \left(\frac{10}{19}\right)^{k-1}
  \text{.}
  \label{eq:k-wall_prob_bound}
\end{equation}
This shows that forming large width walls are exponentially unlikely and, in particular, $k>2$ walls are one order of magnitude less likely than shorter ones. We expect that as $k$ increases even this bound won't be saturated as consistent assignments become increasingly more constrained. The probabilities calculated in this section will be utilised in Section \ref{ssec:stability} to consider the effect of perturbations and to calculate the typical length-scales over which operators delocalise.

% Naturally, for $k=1$, $Q_c = \mathds{1} \otimes \sigma_c \otimes \mathds{1}$ by virtue of \Cref{lemma:conserved_pauli}.

%\chris{Integrate this into the k wall section.}
We have constructed infinite families of non-interfering $k$-walls of both the $\mathrm{SWAP}$ and $\mathrm{FSWAP}$ kind which host $k$-local conserved quantities. These are formed from the oscillations of a local charge within the wall subspace. For the interfering families, we have found instances in which there is no $k$-local conserved quantity restricted to the central subspace. As a result, the existence of a $k$-wall doesn't imply the existence of a $k$-local conserved charges while a general $k$-wall (containing both $\mathrm{SWAP}$ and $\mathrm{FSWAP}$-like gates) may host conserved quantities which are not constructed from a local operator's orbit. We elaborate on these instances in Appendix \ref{app:fswap_consistency}.
\begin{comment}
\chris{We conjecture that the existence of $k$-walls imply the existence of $k$-local conserved operator.
This is true in all the infinite families of $k$-walls which we have managed to construct, including interfering families, which we won't discuss here.
Nevertheless, we note that }
\chris{The interference-free walls for larger $k$ have simple conserved quantities...}
\end{comment}

Deciding whether an $n$-qubit disorder instance of the model contains a $k$-wall (and a $k$-local conservation law) is straightforward at least in the computational sense. This is due to the Clifford nature of the circuit, in which the dynamics can be represented in a $2n$-dimensional phase space as noted in Section \ref{sec:math_prelims}. Therefore, finding invariant subspaces in the circuit instance is a polynomial time complexity problem which can be performed by a primary decomposition of the symplectic image of the unitary operator of the circuit \cite{Steel1997}.

\begin{comment}
    \color{red}\begin{itemize}
  \item  Mention other applications of the consistency counting method?
  \item Mention connection to closed timelike curves?
  \item Mention applicability for qudits or save that for conclusion section?
\end{itemize} \color{black}
\end{comment}

\subsection{Stable localisation in \texorpdfstring{$p<1$}{p<1} model as fragmentation}
\label{ssec:stability}
% \color{red} Is there operator backflow in Clifford? YES BUT UNLIKELY. Can we put a lower bound too? YES.\color{black}
% In the following section, we will consider the gate configurations in the $p=0$ model that generate localised operators and then study the stability of localisation against perturbations for $p>0$. 

\begin{figure}
    \centering
    \includegraphics[width=\linewidth]{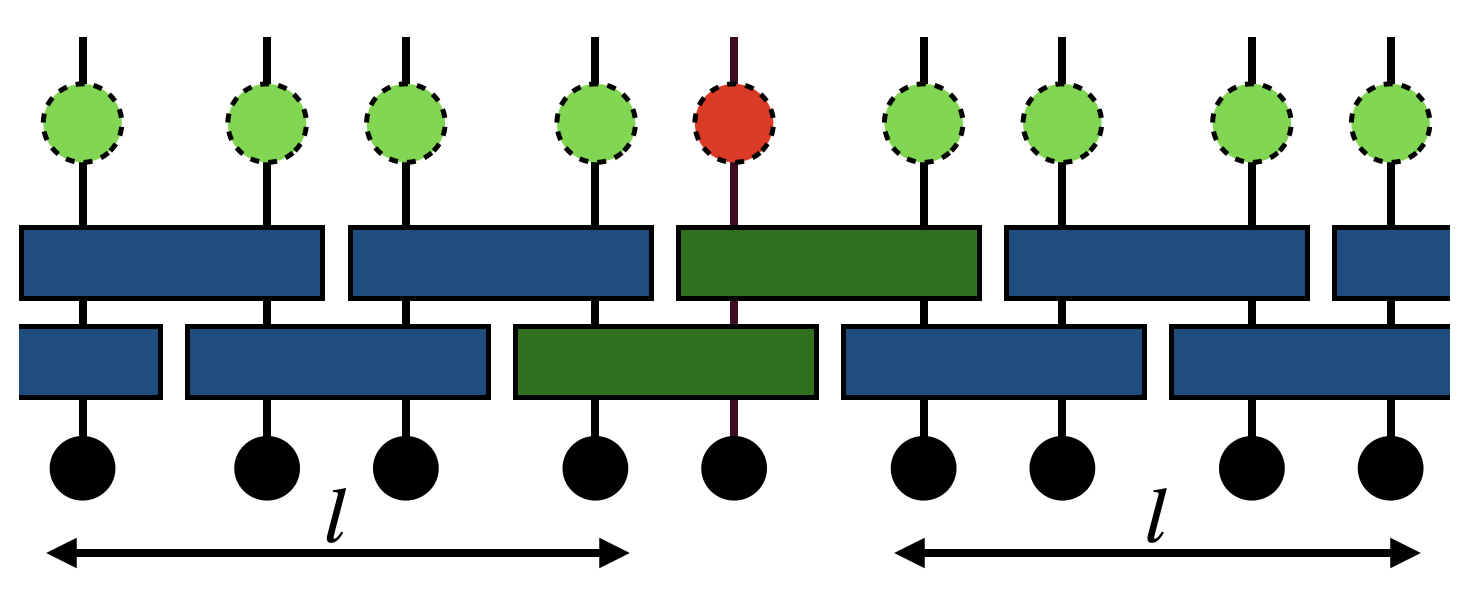}
    \caption{\small 
    %A symmetric region of the circuit, with exactly one $1$-wall on the central qubit (shown in green). Due to fragmentation of the Pauli space, locations for perturbations shown in light green act within invariant subspaces and therefore localisation of operators supported on the left and right is preserved. The red gate potentially disrupts the wall by mixing subspaces, inducing delocalisation of operators and restores ergodicity in the circuit for late times.
    Circuit whose unperturbed part consists of exactly one 1-wall at the central qubit: green (blue) rectangles represent Clifford gates in a 1-wall (non-wall) configuration. Arbitrary single qubit perturbations can be added at the light-green sites to preserve non-ergodicity. The red single-qubit gate potentially disrupts the wall by mixing subspaces, inducing delocalisation of operators and restores ergodicity in the circuit for late times.}
    \label{fig:left_right_subspaces}
\end{figure}
In this section, we embed walls in to larger circuits and study the subspaces spanned by them in relation to the stability of localisation. We develop the idea of fragmentation of a random circuit which will be the basis of understanding the role of perturbations in the model and its implications for the dynamical features of localised/ergodic phases. Consider a circuit region of $2l+1$ qubits bi-partitioned by $1$-wall, as in Figure \ref{fig:left_right_subspaces}. Let us define the subspace of left-conserved and right-conserved ($2l+1$)-qubit operators of the form:
\begin{align}
    \mathcal{L} &= \{\mathcal{\Bar{P}}_l \otimes \{\mathds{1}, \sigma_{\mathrm{c}} \}\otimes \mathds{1}\} \backslash \mathds{1},\label{eq:left_inv_subspaces}\\
    \mathcal{R} &= \{\mathds{1} \otimes \{\mathds{1}, \sigma_{\mathrm{c}} \}\ \otimes \mathcal{\Bar{P}}_l  \}\backslash \mathds{1}, \label{eq:right_inv_subspaces}
\end{align}
where $\sigma_c$ denotes the conserved Pauli operator on the central qubit as a consequence of \Cref{lemma:conserved_pauli}. We will refer to these spatially localised subspaces as fragments. Since $\mathcal{\Bar{P}}_l$ forms an operator basis, the following holds for arbitrary left/right localised operators $Q_{\mathrm{L}} \in \text{span} \mathcal{L} $ and $Q_{\mathrm{R}} \in \text{span} \mathcal{R} $:
\begin{align}
    U^t Q_{\mathrm{L}} U^{-t} \in \text{span}  \mathcal{L} \text{ for all } t>1, \\
    U^t Q_{\mathrm{R}} U^{-t} \in \text{span}  \mathcal{R} \text{ for all } t>1.
\end{align}

Consequently, any perturbation in the subspace spanned by operators in $\mathcal{L}$ preserves the localised support of an operator since it only mixes operators within the localised subspace of the wall and similarly for the right conserved subspace. This is illustrated on Figure \ref{fig:left_right_subspaces}. Perturbations which mix operators between $\mathcal{L}$ (or $\mathcal{R}$) and the rest of operator space will de-stabilise the wall and allow the spreading of support outside of these subspaces (and across the qubit chain). We will only be concerned with local random rotation perturbations as a minimal example but one may conceive more fine-tuned models where the stability of localisation may be enhanced/decreased. As an example, a prototypical $1$-wall made out of $\CZ$ gates is stable even against $Z$-rotation perturbations on the central qubit which, however, have zero probability in sampling under the Haar measure. This example also shows that walls need not be Clifford gates.

Evolution under single-qubit unitaries creates a superposition of traceless Paulis according to:
\begin{equation}
    R_i q R_i^\dagger  = \alpha_x \sigma_x + \alpha_{y} \sigma_y + \alpha_z\sigma_z,
\end{equation}
where $q$ is an arbitrary traceless operator local to site $i$ and $\alpha_i$ are randomly sampled coefficients. Such a rotation acting at the centre qubit of the wall is able to transform operators outside of $\mathcal{L}$ and $\mathcal{R}$ and thus de-stabilize the wall. Under the Haar-measure, the average amplitude retained on in the conserved subspace after a rotation gate applied on the wall is given by the auto-correlator of the conserved Pauli $\sigma_c$ is
\begin{align}
    % \langle \alpha_L \rangle_\mathds{U} &= \langle 2^{-1}\mathbf{Tr}[R_{\mathbf{n}}(\epsilon) \sigma_c R_{\mathbf{n}}(-\epsilon) \sigma_c] \rangle_{\mathds{U}} \\
    \langle |\alpha_c|^2 \rangle_\mathds{U} &= \left \langle \frac{1}{4} |\mathbf{Tr}[U \sigma_c U^\dagger \sigma_c]|^2 \right \rangle_{U \in \mathds{U}}
    % \whencolumns{\\&}{}= \int_0^{2\pi}\frac{d\epsilon}{2\pi} \int_{0}^{1}dn_L  \left ( \cos^2(\epsilon)+2\sin^2(\epsilon/2)n_L^2 \right)
    % \whencolumns{\nonumber\\&}{}
    = \frac{1}{3}
    \text{.}
    \label{eq:rotation_auto_correlator}
\end{align} This shows that random rotations create, on average, a uniform mixture between local Paulis and thus in typical circuit realisations, rotations will de-stabilise the walls exponentially quickly in time.
%\chris{Check this argument makes any sense.}

A similar analysis establishes that any perturbation within the $k$ qubits between wall edges will break localisation for $k>1$. Although for larger width walls, multiple Paulis may propagate in between wall edges, the $k$-qubit subspace of the inner qubits cannot span the full Pauli group $\mathcal{\Bar{P}}_k$. This constraint however, is necessary broken if non-fine tuned perturbations are applied, otherwise the wall would be reducible. Therefore, the fact that a $k>1$ width walls may have a more complicated conserved subspace does not affect the above stability analysis. In fact, $k$-walls are less stable as $k$ increases as the probability that none of inner qubits receive a perturbations exponentially decreases as $(1-p)^k$.

In the large-system limit, randomly applied perturbations will generate an unperturbed wall in the $p<1$ model as the probability of all walls receiving a perturbation decreases exponentially in system size. As a result, our model exhibits a localised phase for any amount of random variation in applying gates from the $\CZ$-class, similarly to Anderson localisation where infinitesimal disorder localises a one-dimensional system's eigenstates \cite{Anderson1958}. We note that this behaviour gets sharper in the $n \rightarrow \infty$ limit. For $p=1$, $k$-walls of all orders delocalize. In this limit, typically operators reach infinite support as $t\rightarrow\infty$. Our model thus also resembles one-dimensional percolation with a trivial phase transition as $p \rightarrow 1^-$.

\subsection{Typical localisation length of operators \label{ssec:localisation_length}}
\begin{comment}
chris{I've moved this here as I don't think it's needed earlier but might be needed here.}
We can now define the notion of localisation in our circuit model. Let the \textit{range} of an operator $O$, $\mathrm{Range}(O)$, be the smallest set of consecutive qubits on the chain outside of which it acts trivially.

\begin{definition}[Localised operator]
An operator $O$ is localised if there exists a constant $k$ such that $\mathrm{Range} (p(\mathrm{Ad}_U) O) \le k$ for all polynomials $p$.
In other words, there exists some finite ball which contains the range of the time evolved operator for any time.
\end{definition}
It is easily seen that a localised operator retains its total norm in a finite region of the chain at all times.
We note that our circuit is locally interacting hence operator range can grow at most linearly with time (ie. at \enquote{light}-speed rate of two in the circuit), $\mathrm{Range}(O) \leq l+2t$.

\end{comment}

Using the probabilities of wall configurations from Sections \ref{sec:wall_configs}, we are now in a position to estimate the typical localization length of operators. An operator localised within some region of the circuit will spread in either direction during the time evolution until it encounters an unperturbed wall. As a result, the typical localisation length of operators (ie. the volume of a finite region in which they retain their total operator norm) is the typical spacing between unperturbed walls. We estimate this length scale in the following.

We are only going to consider $1$-walls and $2$-walls as their sampling probability is of the same order, whereas higher order walls form with exponentially decreasing probability with $k>2$ walls having one order of magnitude less probability at most. Additionally, higher $k$-walls are more prone to being destabilized. Let us define the stopping probability of an operator in at most $2$ steps due to unperturbed walls:
\begin{equation}
    s(p) = \mathds{P}(\text{$1$-wall})(1-p) + \mathds{P}(\text{$2$-wall})(1-p)^2 + \mathcal{O}(p^3).
\end{equation}
The probability of operator localisation reduces to a Bernoulli process in this approximation. Therefore allowing the propagation of a local operator up to $x+2$ steps from a reference point on the chain has the following distribution (along either direction):
\begin{equation}
    S(x) = (1-s)^{(|x|-1)}s = \frac{s}{1-s} \exp\left (-\frac{|x|}{\mu} \right ),
\end{equation}
which immediately gives the typical distance between walls:
\begin{equation}
    \mu(p) = \frac{1}{|\log(1-s(p))|} \sim \frac{44}{(1-p)}. %\text{ as } p \rightarrow 0^+, 1^- 
    \label{eq:loc_length}
\end{equation}
Therefore, the typical volume of localised fragments are $44$ qubits which is tunable by the probability of applying perturbations. As $p \rightarrow 1^-$, $\mu(p)$ exhibits a divergence indicating that walls may only form at infinity. Note that the Clifford-only ($p=0$) circuit's wall distance is significantly higher than was previously seen in numerical calculations \cite{Farshi2022_2D}. The reason for our increased fragment length is the removal of product Cliffords from the circuit which greatly reduces the wall formation probability.

We note that any $\CZ$-like gate is able to form a wall with another $\CZ$-like gate at arbitrarily length scales provided the consistency conditions are met for the intermediate gates. This means that wall formation is a correlated process across all positions due to the existence of arbitrary width $k$-walls. In the previous calculation, we have ignored this effect, similarly to the method shown in the Appendix of \cite{Farshi2022_2D}. Using the exponential bound of Equation (\ref{eq:k-wall_prob_bound}), this will make a negligible error in estimating fragment sizes. 

\subsection{Fragmentation in the thermodynamic limit}

We previously saw how the presence of a $1$-wall leads to the operator algebra decomposing into invariant subspaces which we called $\mathcal{L}$, $\mathcal{R}$ implying that $\mathcal{LR}$ is also invariant.
We can generalise this idea to an extended system with potentially many $1$-walls.
As before, we'll assume that all the walls in the system are $1$-walls which gives a lower bound on the true number of fragments while the true fragmentation, by taking into account all walls, will be slightly more fine-grained.
Between each of these, operators become trapped producing a space of operators local to this subsystem fragment, and identity outside of it.
In the same manner as for constructing $\mathcal{LR}$ products of these spaces form further invariant operator subspaces.

We can formalise this idea by introducing projection superoperators onto each of these local operator fragments, and observing that these are conserved superoperators.
There will also be invariant subspaces for the conserved Paulis and corresponding projection superoperators.
The conserved superoperators can be put together into a commutant algebra and the invariant subspaces understood through representation theory in the manner of Ref.~\cite{Moudgalya2022}.

In the scaling limit as $n \rightarrow \infty$ and for $p<1$, with high probability, the number of walls is extensive in $n$.
This means that the number of fragment subspaces grows exponentially.
The additional fragmentation due to conserved quantities within walls and longer-range walls also gives an extensive number of projections and does not qualitatively increase the fragmentation beyond the exponential.

We could form a basis of stabiliser states and connect the operator-space fragmentation experienced by the stabiliser generators to fragmentation in a basis of stabiliser states.
In this way, the operator-space fragmentation that we discuss here is perhaps just a Heisenberg picture viewpoint to Hilbert-space fragmentation~\cite{moudgalya2022Scars_HSF}.

\section{Entanglement dynamics \label{sec:entanglement}}

Based on the previously introduced operator space fragmentation, we now turn to investigating its effect on the dynamics of a pure state's entanglement entropy across fragment boundaries. Our basic setup will be to consider how much entanglement is created across walls for an initially unentangled pure state. In particular, we consider bi-partitions of the circuit around $1$-walls (as on Figure \ref{fig:left_right_subspaces}). Since our circuit does not contain product unitaries, we expect that entanglement flows across walls even though operators cannot escape. We analytically show that across localised subsystems, entanglement may only increase by a finite amount at all times since walls host a Pauli subgroup in their intermediate region. We verify our results in finite-size numerical simulations.

\subsection{Walls bound entropy}

\subsubsection{The \texorpdfstring{$p=0$}{p=0} limit}

The Clifford nature of the $p=0$ model naturally suggests to consider the entanglement of stabiliser states as operator spreading is equivalent to entanglement spreading in that case. An $n$-qubit stabiliser state $\ket{\psi}$ is one which fulfils $g_i \ket{\psi} = \ket{\psi}$ for the generators of an Abelian subgroup $\{g_i\}$ of the $n$-qubit Pauli group. A pure state at all times is represented by $n$ independent Pauli strings under multiplication. At time $t$, we write the reduced density matrix of a bi-partite state  $\rho_{L}= \mathbf{Tr}_{{R}}[U^t\ket{\psi}\bra{\psi}U^{-t} ]$. The Von-Neumann entropy is defined as $S^{\text{VN}}(\rho_L) = -\mathbf{Tr}[\rho_L \log_2 \rho_L]$. It was previously shown that stabiliser entanglement takes the following form \cite{Nahum2017}:
\begin{equation}
    S^{\text{VN}}(\rho_L) = |L| - N_L, 
\end{equation}

\noindent where $N_L$ is the number of independent stabilisers in the reduced set $\mathbf{Tr}_{R}[\{ g_i \}]$ traced over subsystem $R$ and $|L|$ is the complement subsystem's size.

Consider the circuit instance on Figure \ref{fig:left_right_subspaces} and take a bi-partition right to the wall, ie. between the central qubit and the next one to the right. By virtue of \Cref{lemma:conserved_pauli}, there exists a subspace of $l$-qubit operators for $x\leq l$ on the chain which is closed under the Floquet evolution of all times. This ensures that one can always choose $l$ independent stabilisers from this subspace, that is, $l\leq N_L$. For our subsystem $L$, $|L| = l+1$ giving an upper bound $N_L \leq l+1$. This yields an entropy bound for all stabiliser states across the wall to be $0\leq S^{\text{VN}}(\rho_L) \leq 1 $ at all times. Therefore, disorder instances of the model containing $1$-walls, entanglement is bounded as disjoint fragments may only share entropy through the wall gates (that act as a \enquote*{bottleneck}). It is worth noting that the above calculation predicts the existence of states which remain unentangled across the wall at all times. In particular, these will be stabiliser states where there are $(l+1)$ independent stabilisers in $L$ for all times which means that the stabilisers of the state are all elements of the left-conserved subspace $\mathcal{L}$. For the prototypical $1$-wall made out of two $\CZ$ gates, such a state will be $\ket{0}^{\otimes 2l+1}$ which is stabilized by operators with a local $Z$ on each site.

\subsubsection{Average entanglement for \texorpdfstring{$p>0$}{p>0}}

We argue that the above bound persists across walls which have no perturbation on the central qubit even when $p>0$. Starting from a stabiliser state, the single-particle perturbations generate a superposition of stabilisers similarly to a branching process. If the peturbations act within the conserved subspace $\mathcal{L}$, all stabiliser generators across the branches of the evolved state are either $\{\mathds{1}, \sigma_c\}$ on the central qubit (upon tracing out the right subsystem). In fact, one can choose a generator set such that only one Pauli has non-trivial support on the central qubit for each stabiliser state. The entanglement spreading is then controlled by the spreading of this operator under the wall evolution which is determined whether its local support is $\sigma_c$ or not. Again this yields $0\leq S^{\text{VN}}(\rho_L) \leq 1$ at all times, even though the circuit evolution creates an exponentially large superposition of stabilisers. Therefore, in the $p<1$ model, there are are still weakly entangled bi-partitions. 

Consider now the average entanglement $\langle S^{VN}(\rho_L) \rangle$ of the circuit over the disordered Clifford ensemble. Although the earlier entanglement bound holds for any circuit instance with an unperturbed wall, one can make a tighter estimate for the ensemble average. Assuming that a sufficiently random ensemble on the single-particle rotations creates an equi-partition of stabiliser states across samples, $1/3$ of which will have $\sigma_c$ on the central qubit (and therefore remains unentangled) and the rest will share a single unit of entropy. We thus expect that the steady-state (\enquote*{thermalised}) entanglement will follow $\langle S^{\mathrm{VN}}(\rho_L) \rangle \leq 2/3$.

 We only consider $1$-walls in our analysis as higher order walls are more prone to being destabilised even in the $p<1$ model as \textit{any} perturbation breaks the localisation conditions when it acts on the wall's intermediate qubits. By a similar bi-partition of the circuit across a higher order $k$-wall, we expect a similar, but less stringent bound to hold on entanglement entropy as the inner subspace of the wall may have more independent stabilisers (e.g. $4$ in the case of a $2$-wall) which would give a looser bound although still rendering the fragments weakly entangled considering typical fragment sizes are on the order of tens of qubits when $p \ll 1$.

For bi-partitions within fragments, one expects that the random perturbations generate a chaotic state evolution and thus volume law entanglement with increasing subsystem size. For perturbations that de-stabilize walls, there exists an average timescale over which the fragment subspaces mix due to the rotations and therefore there is an early-time signature of localisation. This timescale depends on the probability distribution of perturbations in the following way. Considering a probability measure biased towards the identity will require several Floquet iterations until the perturbations take effect so that significant operator support is transferred onto delocalised Pauli operators (as measured, for example, in operator norm). Our choice of Haar-random perturbations create an equipartition of Paulis therefore leading to rapid mixing between subspaces and therefore shortening the transient timescale of mixing between localised/delocalised subspaces. 

\begin{comment}
\color{red}
\subsection{Entanglement bounds for larger \texorpdfstring{$k$}{k}-walls}

The entanglement argument may not apply to arbitrary width $k$-walls as the conserved subspaces is more complicated in those. For the intereference-free cases, we expect that maybe a less tight bound holds on the entropy. For example a $\SWAP$-like wall's entropy may be bounded by $4$ instead of one by looking at the number of Pauli generators that can occupy the $k$-qubit inner wall subspace. I'm unsure about the interfering $\FSWAP$-like walls (see Appendix \ref{app:fswap_consistency}). 

\color{black}
\end{comment}

 \subsection{Entanglement signature of \texorpdfstring{$1$}{1}-walls}

\subsubsection{Circuit regions \label{sec:entanglement_numerics}}

In this section, we present exact diagonalisation results for finite-size circuits to study the spreading of entanglement across wall with and without perturbations. The lack of symmetries in random unitaries limits the numerically attainable system sizes in exact diagonalisation to small qubit numbers. We resort to exact methods to be able to probe the long-time thermalisation dynamics as well as spectral statistics of the circuit. Based on the localisation length estimate in Equation (\ref{eq:loc_length}), randomly sampling the non-separable Clifford group would only show localisation on $n \sim 40$, thus one cannot simulate the model exactly without biasing the sampling. To circumvent this, we instead focus on different regions expected to occur in the large system limit, as detailed below. We take the finite-size circuit of Figure \ref{fig:left_right_subspaces} in the following setups:

\begin{itemize}
    \item \textit{Localisation}. A randomly sampled $1$-wall is placed at half-chain such that the centre receives no-perturbation. With probability $p<1$, the remaining qubits receive a random rotation perturbation. For $p=1$, all qubits except the centre receives a perturbation in the spirit of the $p\rightarrow1^-$ limit.
    \item \textit{Perturbed wall}. A randomly sampled $1$-wall is placed at half-chain such that the centre receives a perturbation in all realisations. For probability $p>0$, the remaining qubits receive a random rotation perturbation. For $p=0$, no qubits except the centre receive a perturbation in the spirit of the $p\rightarrow 0^+$ limit.
    \item \textit{Transport}. None of the neighbouring Clifford pairs are $1$-walls. Rotations are applied with probability $p$ throughout the chain. These instances represent the \enquote*{bulk} system away from wall edges that equilibrate at the maximum rate.
\end{itemize}

\noindent Note that for an even number of qubits $n = 2l$, we take an equal bi-partition of the circuit while for $n = 2l+1$ we let the left subsystem to have $l+1$ qubits. The brickwork is truncated by applying single-qubit Cliffords at the edges. We calculate the ensemble average bi-partite entanglement $\langle S^{\text{VN}} \rangle $ of the evolved states $\ket{\psi(t)} = U^t \ket{0}^{\otimes n}$  over the disorder realisations. We define the sample-to-sample fluctuations of entropy, 
\begin{equation} \delta S^{\text{VN}} = \sqrt{\langle (S^{\text{VN}})^2 \rangle - \langle S^{\text{VN}} \rangle ^2},
\end{equation}
\noindent to characterise the entanglement distribution more accurately and to probe the localised subspaces in the dynamics. We focus on the cases $p=0.5$ and $p=1$ as the limited system sizes will not allow us to meaningfully differentiate dynamical features for instances where the probabilities of perturbations are close together.
\begin{figure}
    \centering
    \includegraphics[width=\linewidth]{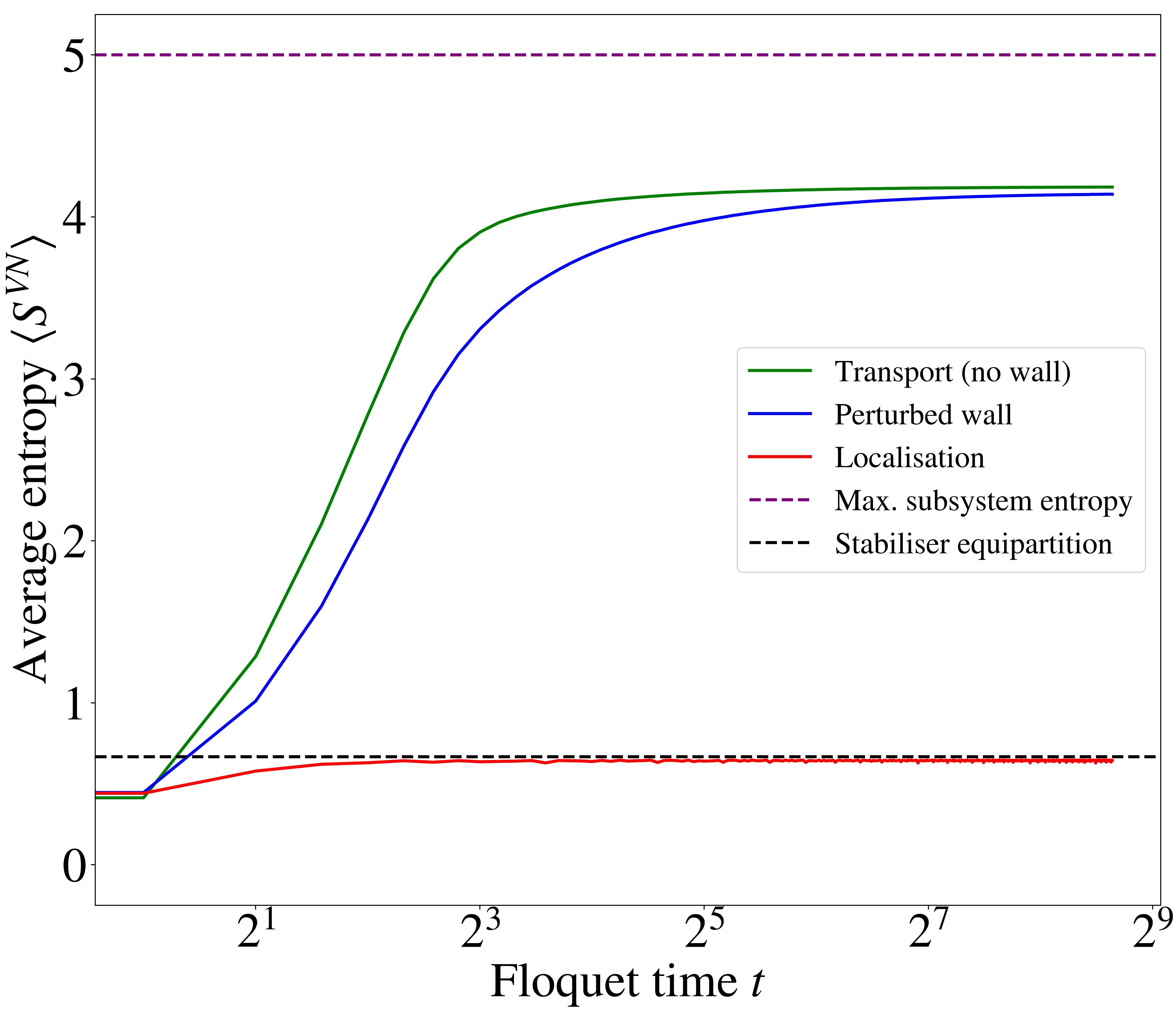}
    \includegraphics[width=\linewidth]{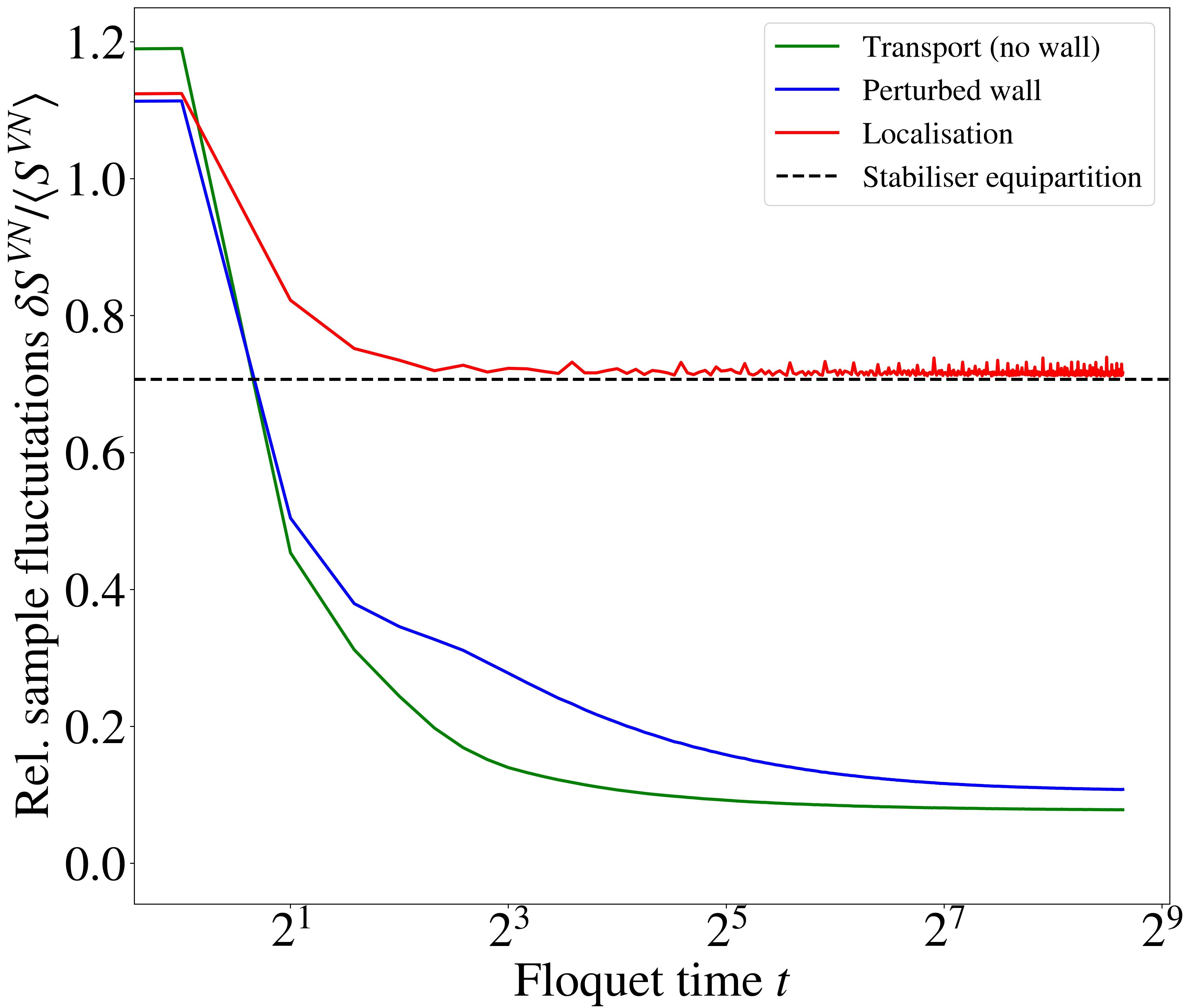}

    \caption[width=\textwidth]{\textit{(Top)} Average entanglement entropy across $1$-walls in an $n=10$ qubit chain for $p=0.5$, averaged over $10^4$ disorder realisations. The $1$-wall in this circuit limits the entropy growth to a constant amount for localised instances in accordance with the theoretical bound. Under rotation perturbation, stabiliser states escape the localised subspace of the wall and share near-maximal entanglement across the wall. \textit{(Bottom)} Large fluctuations of entropy probe localisation for unperturbed wall reminiscent from the discrete entanglement production due to Cliffords.
    %\chris{Data in this figure was calculated with a slightly biased distribution of non-clifford gates and will be recalculated.}
    }
    \label{fig:L11_entropies}
\end{figure}
\begin{figure}
    \centering
    \includegraphics[width=1\linewidth]{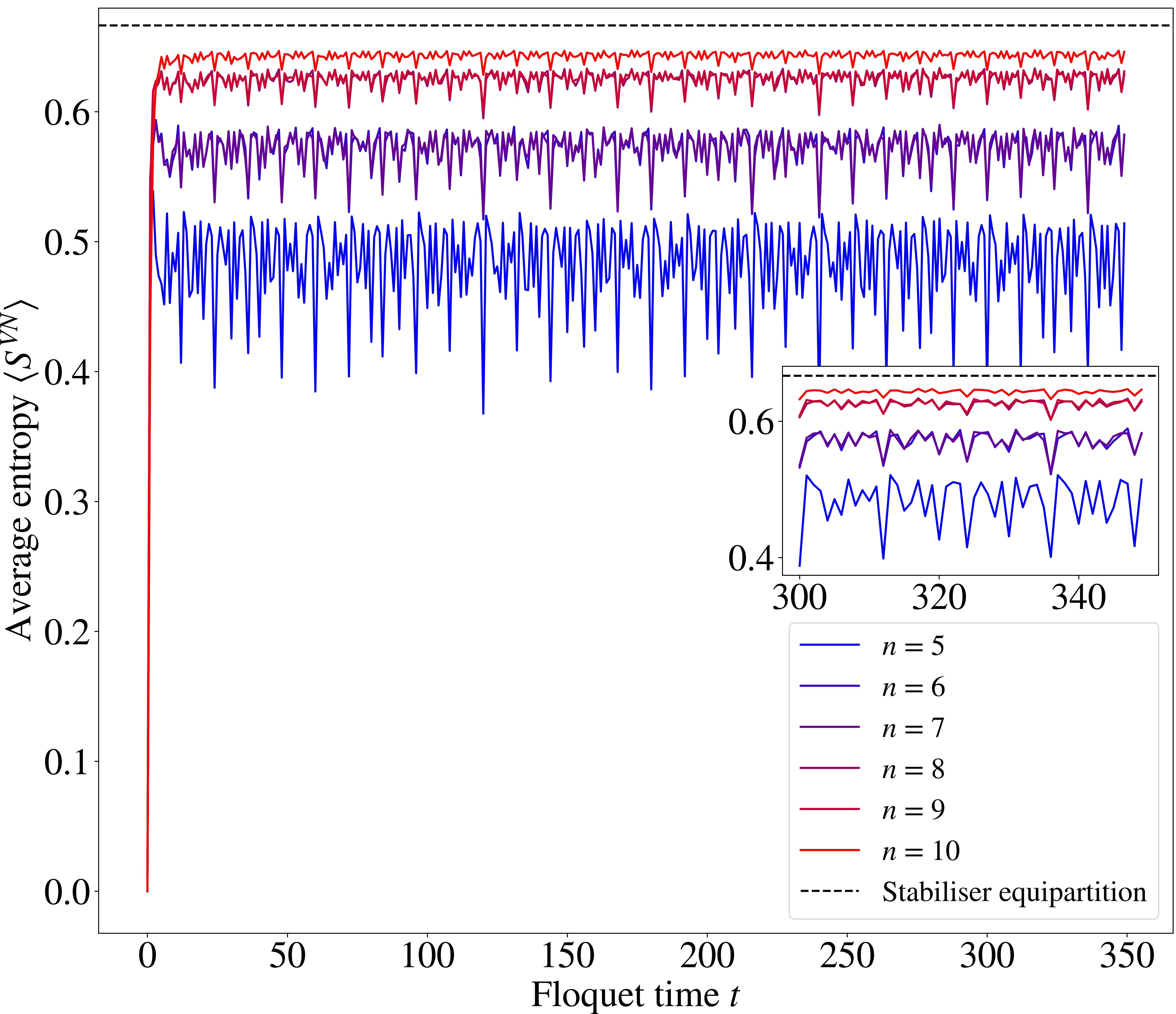}
    \includegraphics[width=1\linewidth]{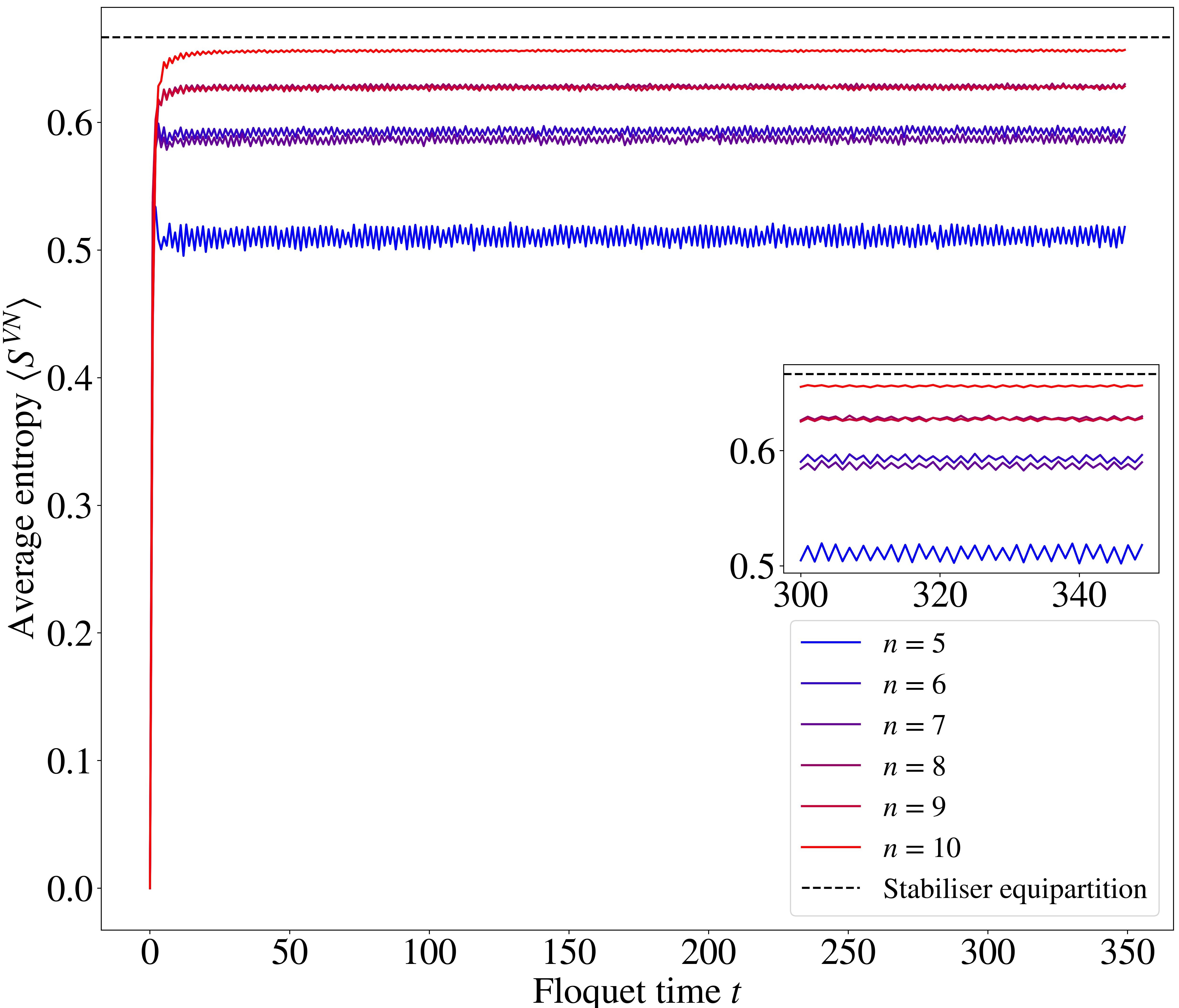}
    \caption[width=\textwidth]{Scaling of entanglement entropy for an $n=10$ localizing qubit chain for $p=0.5$ \textit{(top)} and $p=1$ \textit{(bottom)}, averaged over $10^4$ disorder realisations. As $n$ increases, the entropy approaches the equilibrium value expected for an equi-partition of stabiliser states, two-thirds of which generate maximal entropy (of one bit) while states in the wall's localised subspace have no entanglement.
    %\chris{Data in this figure was calculated with a slightly biased distribution of non-clifford gates and will be recalculated.}
    }
    \label{fig:entropy_scaling}
\end{figure}
\subsubsection{Entropy growth across the wall}

\begin{comment}
    Consider putting the two plots together for entropies!
\end{comment}
We show the time-evolution of average entanglement entropy and entropy fluctuations in Figure \ref{fig:L11_entropies} for $p=0.5$. For an unperturbed wall, the bound based on stabiliser equipartition is corroborated reflecting the fact that circuit disorder generate walls with $X, Y, Z$ conserved subspaces in equal numbers on average and therefore the stabilisers of the initial states anti-commute with the conserved Pauli in $2/3$ of the realisations. Clearly this is consistent with the unit entropy bound for any disorder realisation presented in the previous section.
In the transport limit without any walls, the entropy initially grows then saturates due to the finite system size. Although the growth rate is hard to accurately extract due to finite size effects, we expect at maximum ballistic scaling $S \sim t$ to hold for early times due to the chaotic evolution respecting the locality of the circuit. 
\begin{comment}
\textcolor{red}{AP: $x$-axis in figure 8 is on log scale. "Transport" plot looks linear on this scale. Isn't this different from ballistic. If you want to argue it is ballistic, isn't it better to plot on a linear scale?} 
\end{comment}
Notably, the entropy does not saturate to the maximum subsystem entropy as one would expect from volume-law scaling, although it is consistently bounded by the Page-corrected entanglement entropy of uniform random states on the Bloch sphere. The lack of maximum entropy is arising from the limited randomness of our model indicative that the circuit unitary may only approximate a Haar-random unitary.

For the case of perturbed walls, the growth rate of entropy is reduced from the instances without walls. This is due to the fact that delocalisation is not instantenous, ie. several iterations of the Haar-random perturbation are needed until an appreciable number of operators escape the localised subspace. Using a different measure for the single-particle rotations, the timescale for coupling localised subspaces can be elongated, in essence by biasing the probability distributions towards identity, that we have verified (results not shown). As $t \rightarrow \infty$, the average entropy of perturbed wall instances approaches that of the transport instances indicating that at late times, the effect of fragmentation vanishes as subspaces mix. We expect the gap between these curves to be exponentially small with increasing system size.
\subsubsection{Fluctuations in the entropy distribution}
 Localisation is also observable from the large sample fluctuations of the entropy on Figure \ref{fig:L11_entropies}. For an unperturbed wall, entropy may only be generated across fragments through the Clifford wall gates, therefore the entropy of the state will be between $0$ and $1$ therefore producing large fluctuations in the disorder ensemble. 
 The resulting fluctuations are therefore reminiscent of a binomial stochastic process with success probability $p_s = 2/3$. This is the expected fraction of instances where the wall charge is not $Z$ and therefore non-zero entropy is produced with the initial stabiliser state. We expect the relative fluctuations to be \begin{equation} \frac{\delta S^{\mathrm{VN}}}{\left \langle S^{\mathrm{VN}} \right \rangle}
  \approx \sqrt{\frac{1-p_s}{p_s} }= \frac{\sqrt{2}}{2}.
 \end{equation}

\noindent This value is shown on Figure \ref{fig:L11_entropies} for $\delta S^{\mathrm{VN}}$ which is corroborated by the numerical results. The deviation from the analytical expectation arises from imperfect mixing of stabiliser states for $p=0.5$ such that an anomalously low number of non-commuting stabilisers are produced on the central qubit. This effect may also arise from rare disorder instances where perturbations are weaker than typical. 

 The plateau of fluctuations in the localised limit is in sharp contrast to the perturbed case and the transport distribution where the exponentially large superposition of stabilisers smoothens out fluctuations without the restriction of entropy to the previous two values. In these circuit instances, the fluctutations decay to a constant set out by the distribution of perturbations. For $p=0.5$, the fluctuations don't reach zero indicating that the Clifford-dynamics is still observable in this limit. For $p=1$, fluctuations decay to zero within the timescales considered $t =  \dim \mathcal{H} /2$ (results not shown).

\subsubsection{Convergence to the localised entropy bound}
We also performed a scaling analysis of reaching the theoretical bound on average entanglement for $p=0.5$ and $p=1$ shown on Figure \ref{fig:entropy_scaling}. We observe the approach to the theoretical bound in steps of two in the system size. This finite size effect is due to our selection of bi-partite subsystems. As the minimum subsystem size upper bounds the Von-Neumann entropy, subsystems of $l+1$ and $l$ qubits have close together entanglement, e.g. for $n=6, 7$. The analytical bound is corroborated by all our results supporting the stabiliser equipartition hypothesis generated by the random perturbations. By increasing the system size, steady-state entanglement approaches the bound exponentially with factors of $1/2$. With the addition of a qubit, the size of Pauli space grows by a factor of $4$ but dynamics has to respect the wall constraint on the centre qubit which restricts the growth of effective Pauli space to a factor of $2$. 

 We remark that even with high disorder averaging, time-to-time fluctuations are observable in the steady state entanglement with a period of $2$ Floquet cycles, across all system sizes considered with decreasing amplitude as $n$ increases. This effect we attribute to rare disorder instances in which the gates adjacent to the wall are anomalously close to identity (e.g. additional $\mathrm{CZ}$-like gates with near-identity rotations) that generate oscillations of operators in the vicinity of the wall.

 \subsubsection{Limitations}

We didn't control the formation of higher order walls in the current numerical results. As a result, here might be circuit instances which contain localised subspaces we haven't explicitly accounted for. Such rare instances decrease the average entropy, although we expect this to be a secondary effect.
 
Accounting for higher order $k$-walls in the circuit will not change the results of this section qualitatively, ie. that circuit regions are weakly entangled across fragment boundaries. For extended walls, one has to account for a larger invariant subspace when considering entropy which will lead to less stringent bounds on stabiliser entanglement and a different fluctuation profile in the distribution. Nevertheless, $k$-walls may have a richer entanglement dynamics for bi-partitions \emph{within} the wall subspace in the presence of the constraint on Pauli evolution.

The numerical results on entanglement entropy support the view that the circuit disorder leads to weakly entangled fragmentation of the qubit chain which is manifest in the obstruction of information flow across wall configurations (for $p<1$). We emphasise that our model is not separable across any bi-partition (which would imply that states supported across fragments would remain unentangled) yet the localisation of operators corresponds to limited entanglement spreading, even with non-Clifford perturbations.  

The existence of localised subspaces will lead to area-law scaling of entanglement for $p<1$ in the following sense. By using  the size of a finite subystem as an entropy constraint, one may place a probabilistic upper bound on the average entanglement by the average fragment size of the model (see \Cref{ssec:localisation_length}) that is tunable by $p$. In exponentially rare instances with increasing system size (e.g. sampling only dual-unitary Cliffords in the circuit), this bound may be violated. 

For $p=1$ model, we expect volume law scaling that we probed through the perturbed and no-wall instances. In these, finite-size effects are prohibitive in reaching the scaling regime in which volume-law could be argued which is why we have omitted showing these numerical results. To understand the thermalised entanglement of the model in this limit, one would need to look at the interplay of dual-unitary classes in Clifford with the $\mathrm{CZ}$-class more quantitatively. 
\begin{comment}
Calculating the Von-Neumann entropy has proved to be an ineffective probe into the thermal dynamics \textit{within} fragments. As a result, we complement this picture with spectral calculations in the next section which are less sensitive to the small system sizes and able to give us insight to the ergodicity of our circuit in the presence of unstable fragments.

\textcolor{red}{(AP: Last 2 paras of this section can be made more crisp. Maybe state why the fragmentation of k-walls will not qualitatively change the picture, but quantitatively the fluctuations will receive contributions and in particular the finite size effects to volume law. Do the weakly entangled fragments still lead to volume law or can they have lower entanglement?) }
\end{comment}

\section{Spectral probes of chaos \label{sec:spectral_chaos}}
In order to understand the dynamics of the localised fragments beyond the limitations of looking at entropy distributions, this section is dedicated to the spectral properties of the disordered Floquet model. Spectral probes are a useful way of quantifying the (ergodic) properties of equillibrium ensembles as they converge quickly in the number of qubits ~\cite{Wigner1967, Mehta2004}. We show that the localised fragments evolve chaotically under an approximate Haar-random ensemble and study the ramifications of localisation for the form factor fluctuations and the shape of the ramp. We take the Haar-ensemble as a working definition of ergodicity in unitary circuits. Brickwork random circuits of two-qubit Haar-random gates have been recently considered as minimal model to form approximations of the Haar-ensemble in polynomial depth \cite{Brandao2016} as well as (non-periodic) random Clifford circuits with non-Clifford perturbations to generate Haar-random dynamics. 

Our model, however, has several barriers to reaching a thermal ensemble in this sense. First, the brickwork design of Clifford gates greatly limits how well we may approximate random states even without Floquet symmetry -- the $n$-qubit Clifford group is able to form up to an exact unitary $3$-design (ie. it can produce a discrete ensemble which have up to moments equal to the Haar-ensemble on three replicas) \cite{Huangjun2017, Haferkamp2022}. In the perturbed model, although this barrier is lifted, the Floquet nature of our circuit provably prevents the convergence of the Haar-ensemble if one averages over time as the recurrence of quasi-energy states is incompatible with the random matrix spectra of the CUE ensemble \cite{Roberts2017, Pilatowsky-Cameo2024_CHSE_CUE}. In our calculations, we ignore this subtlety and focus instead on the spectral correlations within the first, randomly sampled, layer of the circuit -- for which it is meaningful to quantify the emergence of chaos by comparing with CUE.

\subsection{Spectral form factor \& fluctuations}

\begin{figure*}
    \centering
    \includegraphics[width=0.49\textwidth]{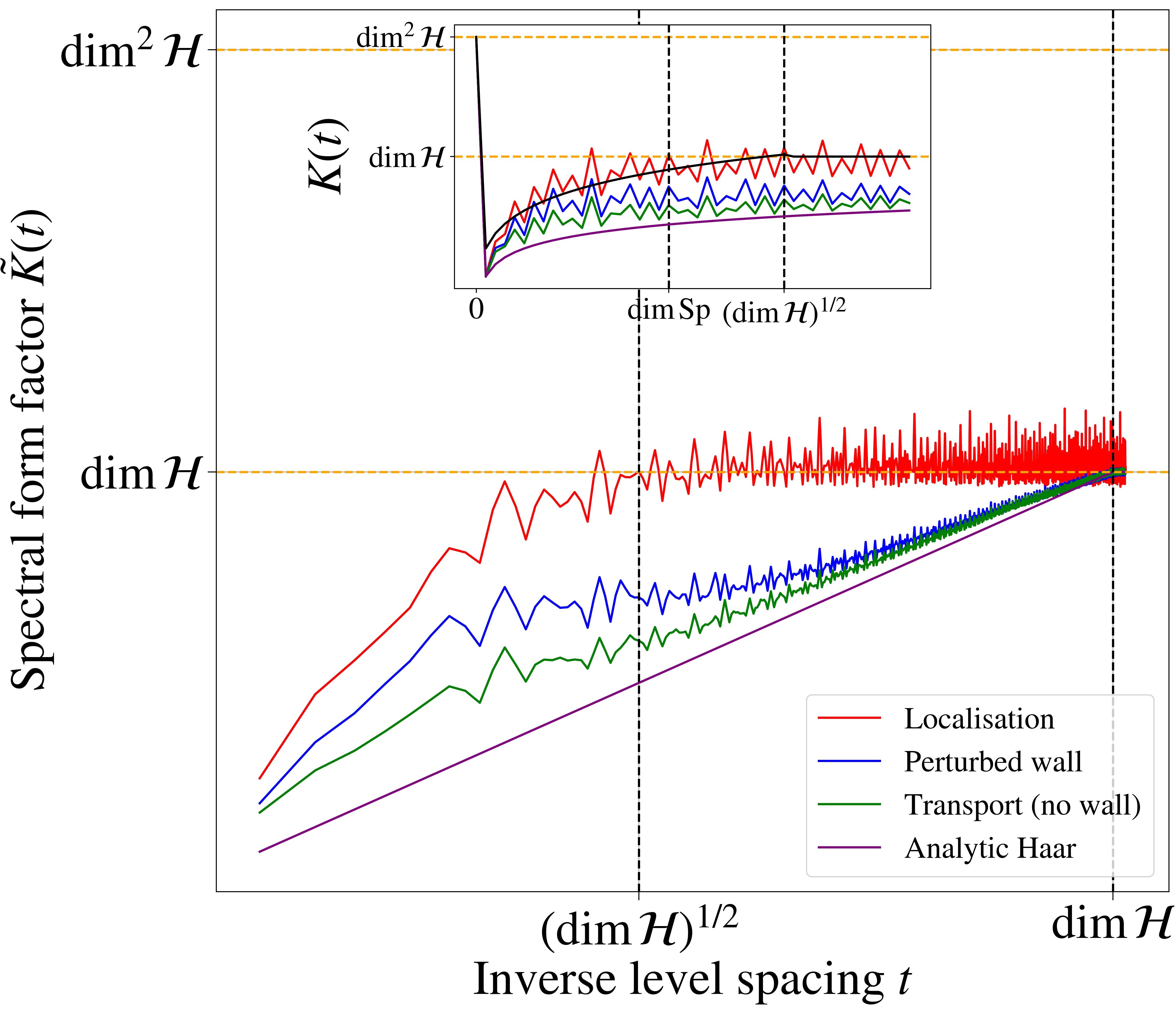}
    \includegraphics[width=0.49\textwidth]{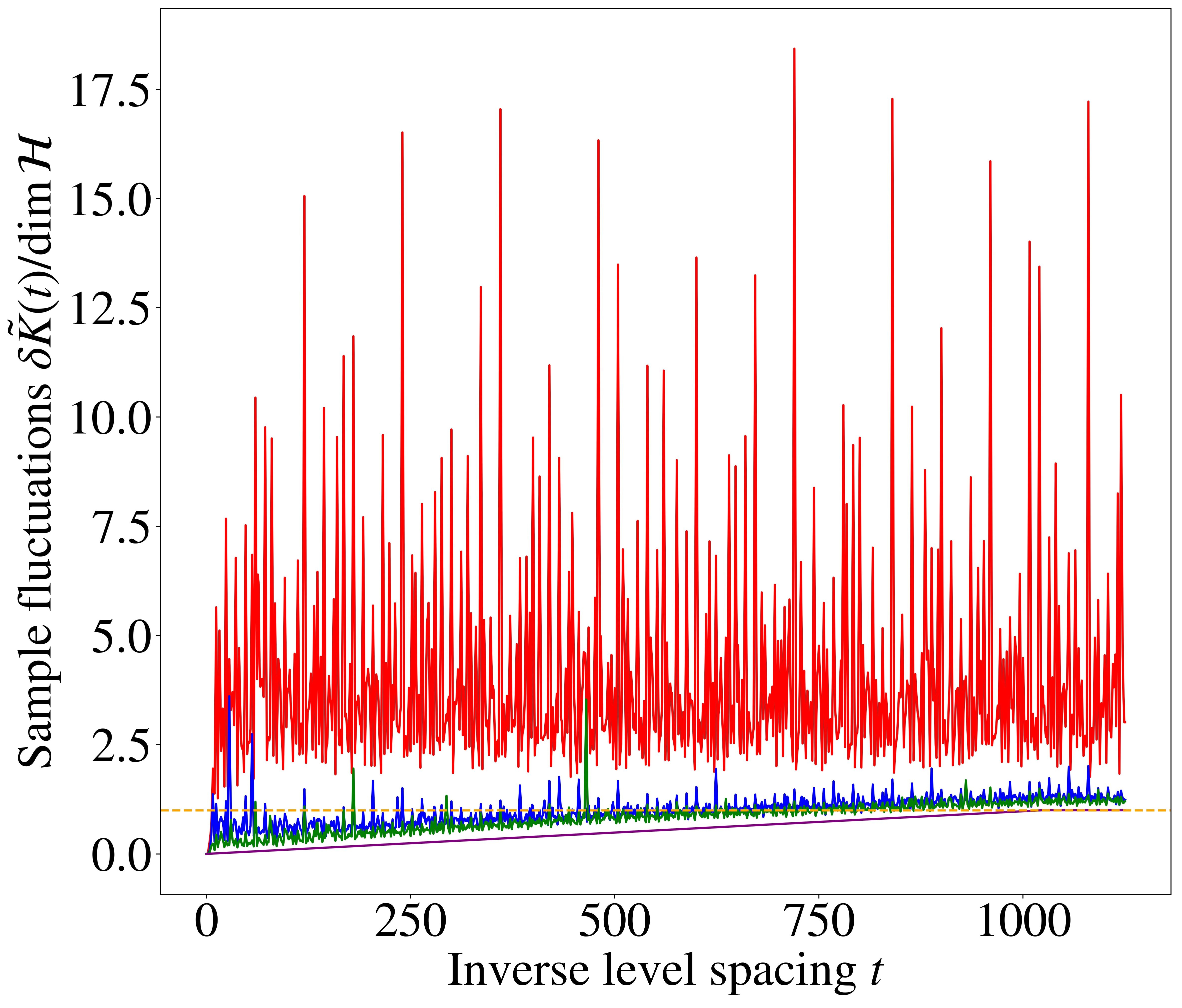}
    \includegraphics[width=0.49\textwidth]{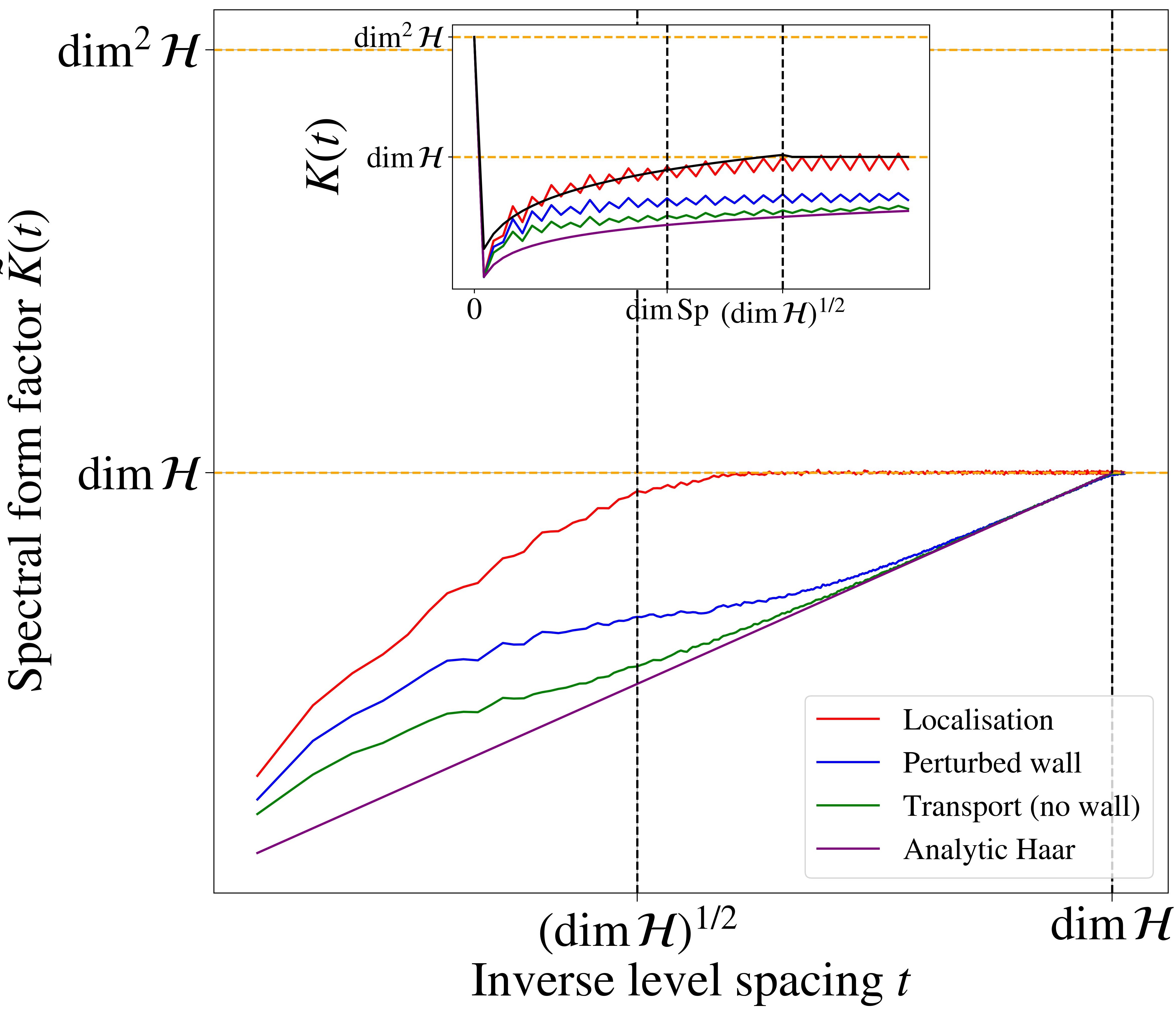}
    \includegraphics[width=0.49\textwidth]{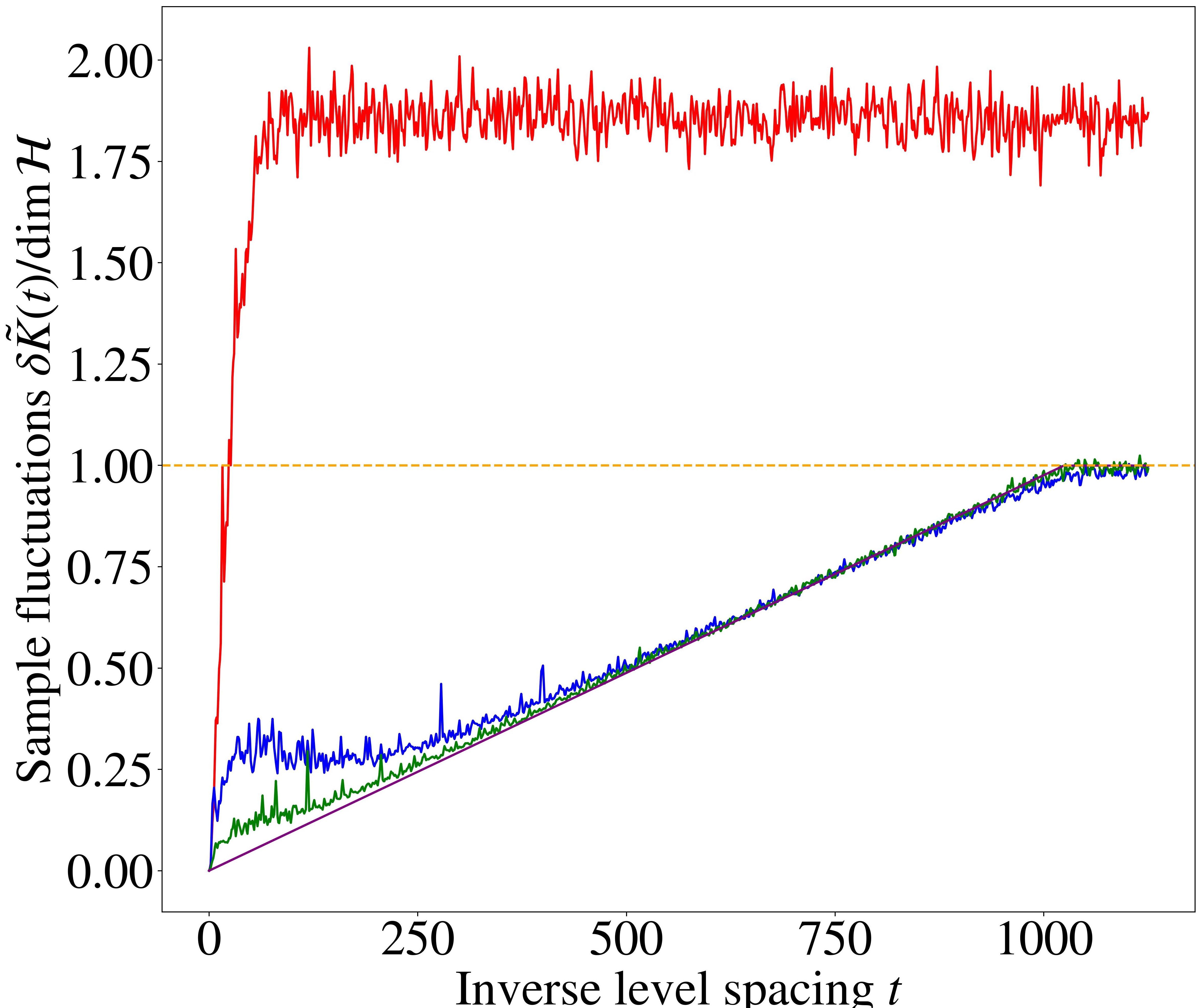}
    \caption{\small Smeared spectral form factor $\Tilde{K}(t)$ \textit{(left)} and sample fluctuations $\delta  \Tilde{K}(t)$ \textit{(right)} for $n=10$ qubits in the $p = 0.5$ model \textit{(top)} and $p=1$ (\textit{bottom}). Main figures are log-log scale while the left insets are semi-log scale in $\Tilde{K}(t)$ showing early time evolution of the form factor without Gaussian smearing. Data averaged over $10^4$ realisations and shown with Gaussian smearing with time window $\Delta t = 1$. Inset shows early time ramp of $K(t)$. Comparison is shown with the analytical curves from Equations (\ref{eq:sff_haar_avg}) and (\ref{eq:sff_haar_std}). The \textit{Fragmentation} curve shows the ansatz $t^2 + 4t$ (see main text) while $\dim \mathrm{Sp}=2n$ is the dimension of the symplectic phase space.
    %\chris{Data in this figure was calculated with a slightly biased distribution of non-clifford gates and will be recalculated.}
    }
    \label{fig:SFF}
\end{figure*}

 Our main probe will be the \textit{spectral form factor}, that we define for an ensemble of unitaries $ \mu$ using the $t$-moments: $K(t) = \langle |\mathbf{Tr}[U^t]|^2\rangle_{\mu}$ where $\langle \cdot \rangle$ denotes an ensemble average. By writing the uni-modular spectrum of  $U\in \mu$ as $\{ e^{i\theta_a}\}$ where $\theta_a \in \mathds{R}$ are real angles, we find:
\begin{equation}
    K(t) = \left \langle \sum_{a, b = 0}^{\dim \mathcal{H}}\exp(i(\theta_a - \theta_b)t)\right \rangle_{\mu}.
    \label{eq:sff_from_spectrum}
\end{equation}
\noindent As the above form suggests,$K(t)$ probes spectral repulsion in an quasi-energy window $\Delta \theta = 1/t$. Another useful representation of the form factor is expressed as the average auto-correlation function of Pauli strings by inserting a complete operator basis in the trace: $K(t) = \left \langle \sum_{P \in \mathcal{\Bar{P}}} \mathbf{Tr}[U^t PU^{-t} P] \right \rangle_{\mu}$. This representation will become particularly useful in understanding the time-to-time fluctuations of $K(t)$ in Clifford evolution. We also calculate the sample-to-sample fluctuations of the form factor defined as: $\delta K(t) = \sqrt{\left \langle |\mathbf{Tr}[U^t]|^4 \right\rangle_{\mu} - K^2(t)}, 
    \label{eq:sff_sigma}$
 that probes $4$-point spectral correlations in the ensemble (repulsion between pairs of eigenvalue pairs). For the circular unitary ensemble, these quantities can be found analytically \cite{Haake1999}:
\begin{align}
    K_{\text{CUE}}(t) &=      \begin{cases} D^2 & \text{if } t = 0 \\
      t & \text{if } t \leq D\\
      D & \text{if } t > D
    \end{cases} \label{eq:sff_haar_avg}   \\
    \delta K_{\text{CUE}}(t) &= \begin{cases}
        t &  \text{ if } t < D/2 \\
        \sqrt{t^2 - 2t + D} & \text{ if } D/2 \leq t < D   \\
        \sqrt{D^2-D}  & \text{ if } D \leq t\label{eq:sff_haar_std}
    \end{cases}
\end{align}
where $D = \dim \mathcal{H}$.
From the above expressions, the \enquote*{dip} of the form factor occurs at $t=1$ indicating that the evolution is characterised by a universal RMT ensemble at all times. Additionally, the Heisenberg time occurs at $\tau = \dim \mathcal{H}$, indicating the timescale over which average level repulsion becomes manifest.

\subsection{Emergence of CUE}

We calculate the spectral form factor and its sample fluctuations for our circuit shown in \Cref{fig:SFF} for $p\,{=}\,0.5$ and $p\,{=}\,1$.
We use the three different circuit setups which probe the spectral behaviour in our circuit between fragments with/without unperturbed walls (\emph{Localisation/Perturbed wall}) and within fragments without walls (\emph{Transport}) as detailed in Section \ref{sec:entanglement_numerics}.
Similarly to Gaussian random models, $K(t)$ displays a dip-ramp-plateau structure although the unitary circuit has a sharp dip at $t=0$. 

In \cite{Farshi2022_2D}, it was found that a Clifford ergodicity is signalled by the plateau time occurring at the dimension of the phase space: $\dim \mathrm{Sp} = 2n$ indicating chaotic dynamics in phase space which also implies that the ramp is exponential in time. In the presence of an unperturbed wall, however, localisation prevents the emergence Clifford chaos due to the existence of fragments. All simulation setups deviate from an exponential ramp shown on the early-time form factor evolution with significant deviations from the linear ramp expected for the many-body chaotic CUE ensemble, too.
 
For the simulation results shown, one expects the ramp to be quadratic in time to leading order which we argue below. Due to fragmentation, we may decompose the Pauli space as a direct sum of localised spaces:
 
\begin{equation}
    \mathcal{\Bar{P}}_n = \mathcal{L} \oplus \mathcal{R} \oplus \left ( \mathcal{L} \mathcal{R} \right) \oplus \left (\mathds{1} \otimes  \{\mathds{1}, \sigma_c\} \otimes \mathds{1}\right) \oplus \mathcal{P}_{\mathrm{rest}},
\end{equation}

\noindent where we have used the left/right invariant subspaces from Equations (\ref{eq:left_inv_subspaces}) and (\ref{eq:right_inv_subspaces}) and $\mathcal{P}_{\mathrm{rest}}$ denotes the remainder of the Pauli space not explicitly shown. The spectral form factor of a product space is the product of form factors within the individual subspaces. Assuming a large density of perturbations, $K(t) \sim t$ within a fragment (at least for early times), the form factor takes a dominant contribution $K(t) \sim t^2$ that modifies the plateau time $\tau_{\mathrm{pl}} \sim \sqrt{\dim \mathcal{H}}$. Accounting for the remaining subspaces, we plot the following ansatz $K(t) = t^2 + 4t$ which is in qualitative agreement with the trend of $K(t)$ for early times, corroborating the quadratic ramp. 

Once a wall is broken, our numerical results show a crossover between a super-linear and linear ramp. This indicates that there is a finite timescale over which fragments couple and the corresponding Pauli subspaces mix. This is likely due to the sampling measure of our perturbations. Deviations from linear ramp are also observable for circuit instances without $1$-walls indicating that the circuit might have higher order walls or invariant subspaces we haven't explicitly constructed. With the existence of multiple, sufficiently chaotic, fragments within the circuit instance, we expect the form factor to be polynomially dependent on time with the degree calculated from the largest product subspace (ie. the number of chaotic fragments). 

The prominent late-time fluctuations of $K(t)$ in time also signal the existence of localisation. This is reminiscent from Clifford dynamics ($p=0$) where non-zero contributions to $K(t)$ come from the recurrences of Pauli operators in the dynamics driving large fluctuations due to the broad distribution of recurrence times for the ensemble of Pauli strings. For $p>0$, this is reduced to approximate recurrences (finite overlap of $U^t P U^{-t}$ and $P$ under the Hilbert-Schmidt norm) although the existence of fragments enhances the rate of recurrences for operators within them. The large temporal fluctuations in form factor seen previously in numerical results in \cite{Farshi2022_2D} are thus smoothened in the perturbed model, although they still signal the lack of phase space chaos of the brickwork Clifford model and the effect of localisation.

We also illustrate the non-ergodicity of the perturbed model inherited from Clifford circuits through local time-averaging. In particular, in the localised model, neither a disorder average or a time-average can reduce a localised Clifford system to the ensemble-average. To show this, we performed Gaussian smearing in a finite time window $\Delta t$ via:

\begin{equation}
    \Tilde{K}(t) = \frac{1}{2}K(t) + \frac{1}{2}\sum_{t' = t-\Delta t}^{t+\Delta t} K(t') N(t'), 
\end{equation}

\noindent where $N(t')$ is a Gaussian random variable drawn with mean $K(t)$ and variance according to the spread of the set $\{ K(t'-\Delta t), ..., K(t'+ \Delta t) \}$. The smeared fluctuations stay consistently above the Haar-value for the localised instances and generally for the $p=0.5$ results. Although $\Tilde{K}(t)$ approaches $\dim \mathcal{H}$ for $p=0.5$, the deviation from the Haar curve in the fluctuations signal that the ensemble is spectrally distinct from CUE in higher order level correlations.  A more detailed quantification of how well the perturbed Clifford ensemble can approximate the CUE ensemble (e.g. through calculating frame potentials \cite{Roberts2017}), we leave for future work. 

In the instances where the wall is either absent or perturbed, $\Tilde{K}(t)$ is consistently above the Haar curve for $p=0.5$. We attribute this to the fact that the limited single-qubit randomness insufficient to approximate CUE in first order level correlations. For $p=1$, $\Tilde{K}(t)$ approaches the Haar curve for sufficiently late times that we may interpret as an emergent Thouless timescale of reaching a chaotic thermal state beyond which the effects of localised subspaces are ignorable. In these cases, the Gaussian smearing is able to decrease the form factor fluctuations at late times to the ensemble average in sharp contrast to the localised case. The numerical results support that the localised phase in our model is a non-ergodic quantum phase while the uniformly perturbed model induces not only an operator delocalisation transition but also the restoration of unitary ergodicity of the circuit ensemble.

\section{Discussion \& Outlook \label{sec:discussion}}

\begin{figure}
    \centering
    \includegraphics[width=\linewidth]{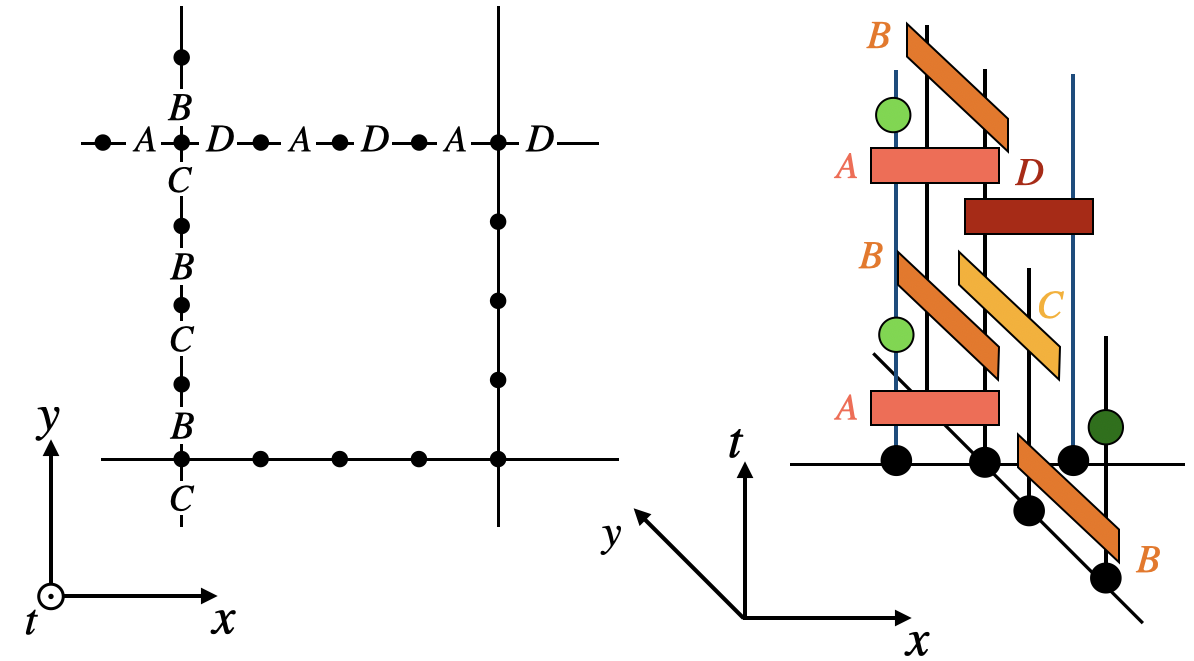}
    \caption{\small 
    We consider a two-dimensional generalisation of the model on a square lattice with $m$ qubits on each edge ($m=4$ in the figure). The Floquet operator is composed of four layers (A,B,C,D) of two-qubit Clifford gates. Hence, the dynamics within each edge is identical to that of the one-dimensional model, where walls can appear with some probability that depends on $m$. This gives rise to a percolation model which produces ergodic or localising dynamics depending on $m$. As in the one-dimensional model, walls can be destroyed by the inclusion of non-Clifford perturbations, represented by circles in the figure.
    %Model extension to two-dimensions, showing a segment of a two-dimensional junction. Percolation transitions can be realised in higher dimensions by evolving each bond on the lattice (or graph) by a one-dimensional Floquet circuit which can generate walls between junctions.
    }
    \label{fig:2D_extension}
\end{figure}

To conclude, we have constructed an interacting random Floquet-Clifford circuit exhibiting operator localisation and studied its stability against generic unitary perturbations. The Floquet symmetry permits strict localisation of the support of operators, even in the interacting case. This is due to the interplay between the locally interacting Clifford circuit controlling the spreading of local operators and the circuit layer of perturbations that scrambles local information within fragments. We formalised the mechanism of operator localisation due to the formation of shallow $k$-walls in the circuit. Furthermore, we showed that these typically lead to emergent conserved charges in the disorder dynamics by studying the invariant subspaces defined by the walls. Our model is thus an analytically tractable limit of interacting, non-integrable quantum dynamics that exhibits features of both Anderson and many-body localisation. 

By looking at the equivalence classes of the Clifford group, we found that localisation arises from the competition between the dual-unitary classes ($\SWAP$ and $\FSWAP$) that generate ballistic (maximal speed) operator spreading and the $\CZ$-class that probabilistically generate invariant subspaces which are spatially localised due to the locality of our circuit. From the $\CZ$-classes, we have also constructed emergent local integrals of motion in our circuit. In our model, localisation still permits large circuit segments to be ergodic within the localised regions. Our exact numerical results provide evidence that a sufficiently large random circuit fragments into weakly entangled regions separated by a sequence of gates forming a wall, which allow a finite but bounded entanglement to spread across the weak link. We have also shown that fragments evolve under an approximate Haar-ensemble and are restricted to reach the CUE ensemble due to the limited randomness of the model. 
\begin{comment}
\textcolor{red}{(AP: Do we need the 2nd part of this discussion? We didn't discuss this in the main text. Can be part of outlook maybe, if required)} Nevertheless, one could construct a similar model with a larger period of more than two layers of randomly applied Clifford and rotations which then are repeated periodically. As long as Floquet symmetry is preserved and local rotations are applied stochastically, our arguments for localisation would still apply, albeit with lower probability as the wall configurations would have more constraints. 
\end{comment}

There are two crucial properties of our model which are necessary for the fragmentation behaviour to be observed: time-periodicity of the circuit (so that arrested spreading persists up to arbitary times) and the possibility to create walls from the gateset used in the circuit. In the case of perturbed Clifford circuits, this is due to the finite probability of sampling the controlled unitary class. Although models without disorder in sampling can also exhibit walls, we opted for a more generic class of Clifford circuits in which the thermodynamic limit is non-ergodic. The general conditions of time-periodic circuits beyond the Clifford group to exhibit walls is an interesting question we leave for future work. We have not considered the generalisation of our model to (non-Clifford) two-qubit perturbations, because this would not change the localisation phenomenology.

%Our results can be generalised to higher spin dimension and higher spatial dimension. 
Our model can be generalised to higher spin dimension using the Clifford group on qudits. Since this alternative model has analogous mathematical structure, we expect Floquet localisation to occur. At the same time, and due to the larger set of equivalence classes of qudit Clifford gates, we expect a richer variety of wall configurations.
%we We expect that Floquet-localisation would still occur in systems with larger local Hilbert space dimension, using qudit generalisations of the Clifford group. These perhaps have a richer variety due to the more complex structure of equivalence classes. 

We consider the generalisation of this model to two-dimensional lattices with different geometries, and show that the chosen geometry determines whether the model has ergodic or localising dynamics. We consider a square lattice with $m$ qubits on each edge, and Floquet operator composed of four layers of two-qubit Clifford gates as illustrated in \Cref{fig:2D_extension}. This interaction pattern ensures that the dynamics within every edge is identical to that of the one-dimensional model, so that walls can appear in some edges, giving rise to a percolation model. The wall probability depends on $m$, so by adjusting this parameter we can engineer a percolation phase transition. In the delocalised phase operators grow indefinitely, and in the fragmented phase operators get stuck within a finite region. The extreme case with $m=1$ is essentially the model introduced in \cite{Farshi2022_2D}, with minor differences, which is in the ergodic phase. Also, as in the one-dimensional model, the walls in the edges of the square lattice can be broken if we perturb the system with non-Clifford single-qubit gates.

Finally, our understanding of localisation rested on the analytical characterisation of $k$-wall gates in Clifford circuits but we emphasise that the notion of a $k$-wall is applicable more generally.
There are some important outstanding questions:
What is the most general class of brickwork circuits exhibiting arrested operator spreading?
What is the relationship between local conserved quantities and operator spreading in circuits?
Beyond circuits, we believe that dynamics of certain topologically ordered Hamiltonians (such as the toric code \cite{Kitaev2003} and $\mathrm{ZXZ}$-model) would also exhibit walls and the associated dynamical features due to the mutually commuting Hamiltonian terms. Understanding the interplay between topological properties and the signatures of dynamical localisation could establish the impact of this work's findings beyond the circuit setting.

\begin{comment}
\color{red} Mention applications of results for quantum computing protocols? Complexity of learning localisation from the dynamics? Benchmarking protocol for determining whether a set of gates generates information spreading or not? \color{black}
 
\end{comment}

\section*{Acknowledgements}

M.D.K.~would like to thank Nicholas Hunter-Jones, Anushya Chandran, Daniel Mark and Saul Pilatowsky-Cameo for insightful discussions. M.D.K.~is supported by the UK Engineering and Physical Sciences Research Council (EPSRC) [grant no.  EP/S021582/1].
C.J.T.~is supported by an EPSRC fellowship (Grant Ref. EP/W005743/1).
A.P.~is funded by the European Research Council (ERC) under the EU's Horizon 2020 research and innovation program via Grant Agreement No. 853368.
Statement of compliance with EPSRC policy framework on research data: this publication is theoretical work that does not require supporting research data.

\bibliography{export}

\clearpage
\appendix

% \onecolumn
% \appendix

% \section{First section of the appendix}
% Quantum allows the usage of appendices.
% If you want your appendices to appear in \texttt{onecolumn} mode but the rest of the
% document in \texttt{twocolumn} mode, you can insert the command
% \texttt{\textbackslash{}onecolumn\textbackslash{}newpage} or just
% \texttt{\textbackslash{}onecolumn} before
% \texttt{\textbackslash{}appendix}.

% \subsection{Subsection}
% Ideally, the command \texttt{\textbackslash{}appendix} should be put before the appendices to get appropriate section numbering.
% The appendices are then numbered alphabetically, with numeric (sub)subsection numbering.
% Equations continue to be numbered sequentially.
% \begin{equation}
%   A \neq B
% \end{equation}
% You are free to change this in case it is more appropriate for your article, but a consistent and unambiguous numbering of sections and equations must be ensured.

% \section{Problems and Bugs}
% In case you encounter problems using the quantumarticle class please analyze the error message carefully and look for help online; \href{http://tex.stackexchange.com/}{http://tex.stackexchange.com/} is an excellent resource.
% If you cannot resolve a problem, please open a bug report in our bug-tracker under \href{https://github.com/quantum-journal/quantum-journal/issues}{https://github.com/quantum-journal/quantum-journal/issues}.
% You can also contact us via email under \href{mailto:latex@quantum-journal.org}{latex@quantum-journal.org}, but it may take significantly longer to get a response.
% In any case, we need the full source of a document that produces the problem and the log file showing the error to help you.

\section{Consistency conditions for \texorpdfstring{$\FSWAP$}{FSWAP}-like 2-walls}
\label{app:fswap_consistency}

In this section, we explicitly construct $\FSWAP$-like two-walls (Figure \ref{fig:fswap_wall_reshaped}.) by looking at assignments of the single-qubit Clifford degrees of freedom that preserve localisation from Section \ref{sec:wall_configs}.
The common feature of wall configurations is that they do not allow the spreading of an arbitrary single-qubit Pauli (and thus the entire Pauli group) within the $k$-qubit inner subspace of the wall.
As discussed in the main text, there are two qualitatively different types of assignments which we call \enquote*{interference-free} and \enquote*{interfering}.
% depending on the subgroup of Paulis that can be assigned to the tensor legs of the wall diagrams, as on Figures \ref{fig:CZ_walls} and \ref{fig:SWAP_like_walls}.

\begin{figure}[htp]
  \centering\includegraphics[width=0.75\linewidth]{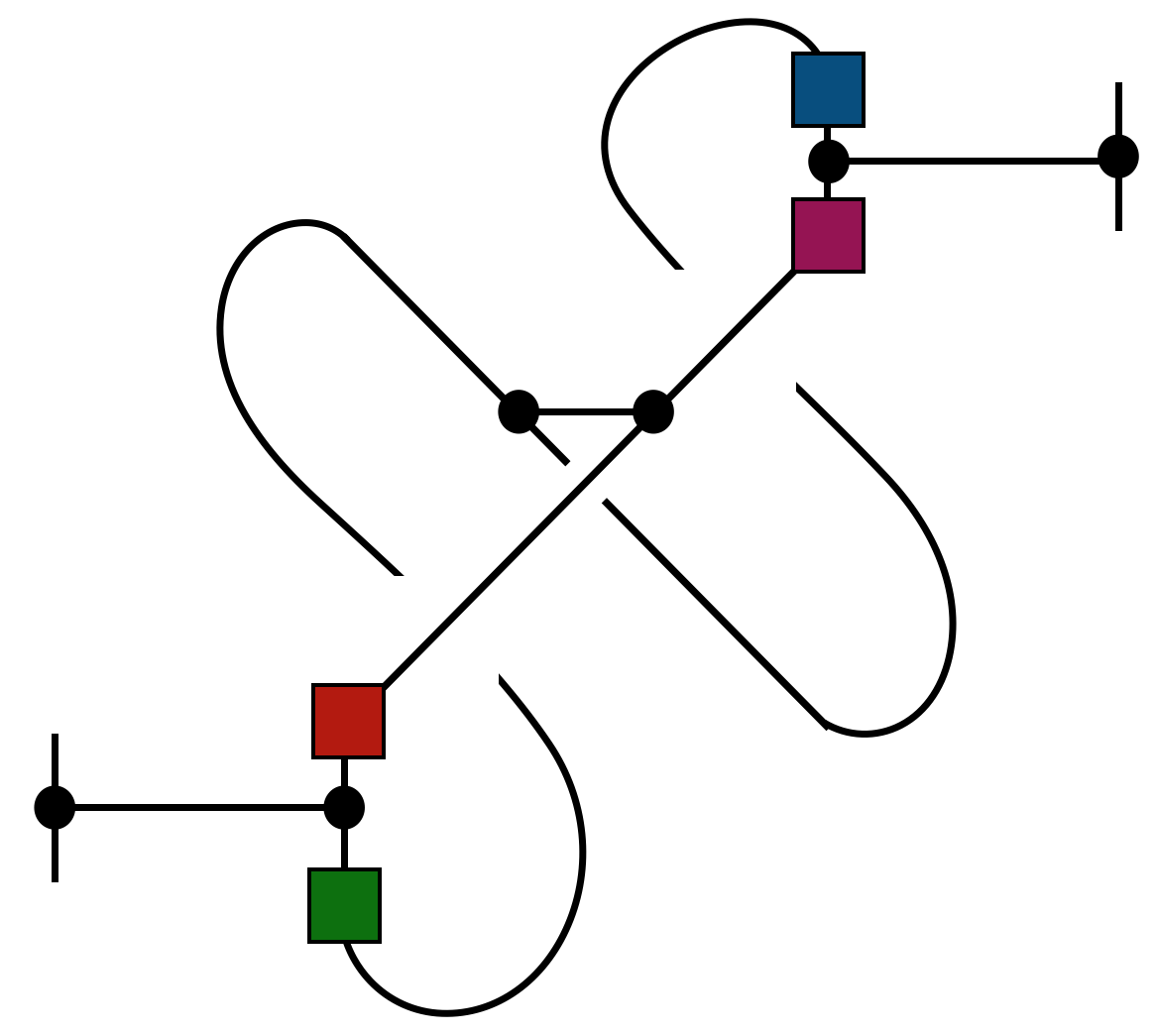}
  \caption{\small General form of $\FSWAP$-like $2$-walls. We have reshaped the tensor such that counting equivalences reduces to consistently assigning Pauli subgroups to the $\FSWAP$ gate's tensor legs.}
\label{fig:fswap_wall_reshaped}
\end{figure}

For the interference-free case, there are $5$ cases shown on Figure \ref{fig:fswap_interference_free}. In these, we ensure that only a single-qubit Pauli subgroup $\{ \mathds{1}, \sigma_c\}$ propagates on each tensor leg. With the width of the wall larger than $1$, however, $\mathcal{C}_1$ elements can change the $\sigma_c$ as the operator spreads through the $\FSWAP$ gates. In these cases, the subgroup of $2$-qubit Paulis that can be generated are given by two local generators.

In Figure \ref{fig:fswap_interference} we illustrate how the interfering case works, and find there are 4 inequivalent choices. One can take the $\mathrm{CNOT}$ gate as the representative $\CZ$-class element to generate interference with $\FSWAP$. Then, we may insert $S$ gates around the target qubit of $\mathrm{CNOT}$-s to change the Pauli inserted into the wall's central subspace from $X$ to $Y$ and vice versa.
In this example, the left and right internal subspaces have the appearance of Jordan-Wigner fermions with parity strings stretching back to the side they were inserted from.
The intersection of these internal subspaces are therefore even collections of fermions such that the strings cancel at the boundary and these form the conserved charges.
% For the interfering case, there are three generators of the $2$-qubit Pauli group that can spread within the inner qubit subspaces as shown on Figure \ref{fig:fswap_interference}. In this case, although not all tensor legs carry a single-qubit Pauli subgroup, the assignment of Paulis ensures that only $\{ \mathds{1}, \sigma_c\}$ reaches the $\CZ$-like gate thereby respecting the localisation constraint.
Therefore, there are in total $9$ assignments for $\FSWAP$-like walls which gives the $1/9$ consistency probability in the expression for sampling probability in \Cref{eq:sampling_prob_fswap}. Constructing interfering assignments for $k>2$ becomes increasingly complicated as one has to look for ever-larger subgroup assignments of the tensor legs that fulfil the localisation constraint at the end of the walls. 

\begin{figure}
  \centering
\includegraphics[width=\linewidth]{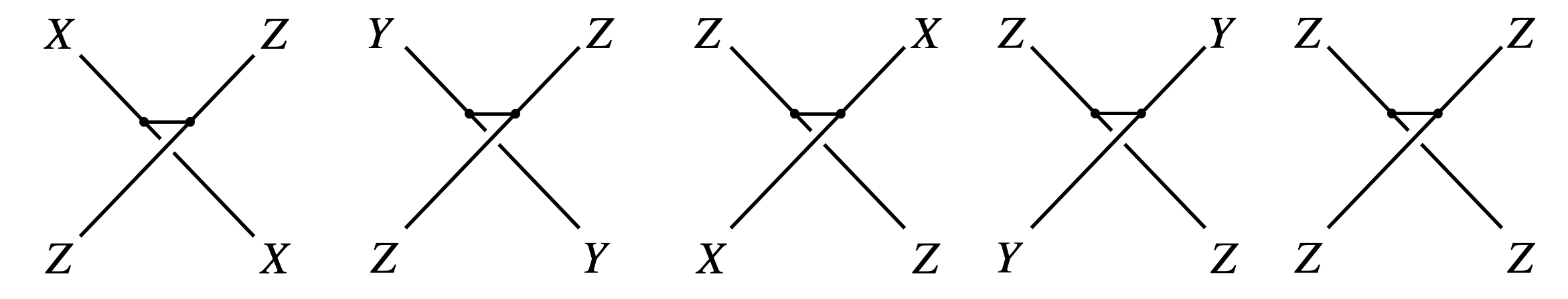}
  \caption{\small Inequivalent single-qubit Clifford assignment for $\FSWAP$-like $2$-walls. In these cases without interference, each tensor leg carries a Pauli subgroup $\{\mathds{1}, \sigma_c\}$. The diagrams should be understood that the Pauli shown on the legs or the identity may propagate in the loops of the wall-diagram.}
\label{fig:fswap_interference_free}
\end{figure}
\begin{figure}
  \centering
  \includegraphics[width=\linewidth]{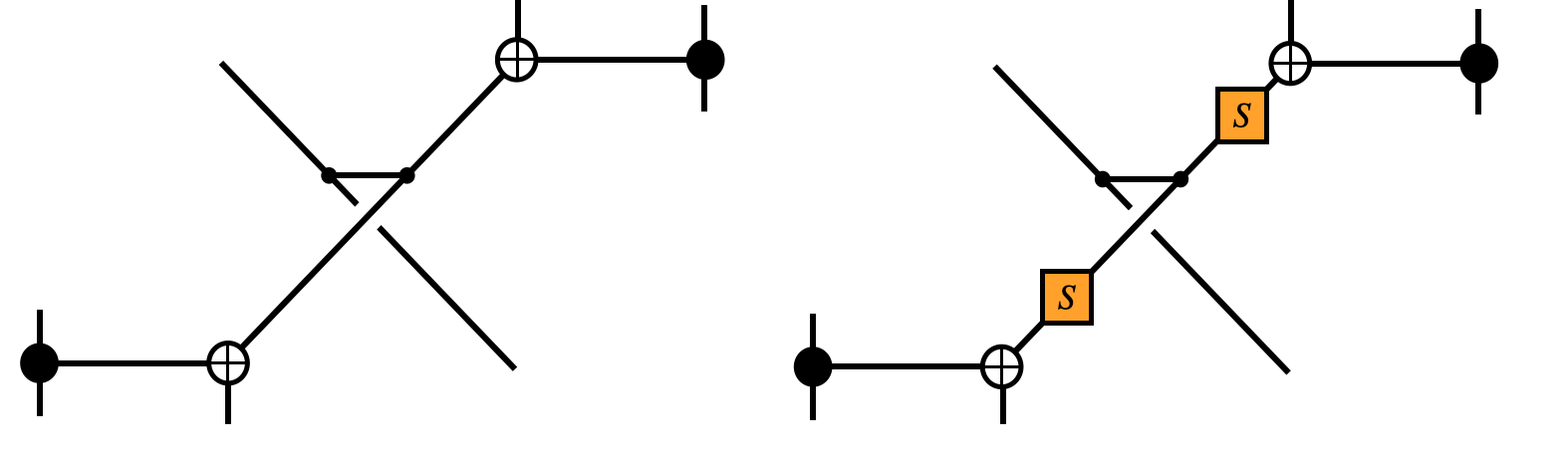}
  \caption{\small Pauli assignment to $2$-walls exhibiting interference. From the left gate, we construct equivalent configurations by applying $S$ gates on the diagonal. $2$ additional instances can be created by adding $S$ gates on the adjacent $\FSWAP$ legs, or on all legs, yielding the $1/4$ consistency probability for this class. Neither of these gates host $2$-local conservation laws.
  }
\label{fig:fswap_interference}
\end{figure}

\section{The 0-walls are exactly the product unitaries}
\label{sec:0-walls}

That product unitaries form 0-walls can be easily and directly verified, so we focus on the other direction of the equivalence.

Without loss of generality, take a complete orthonormal basis of Hermitian operators for each of the Hilbert-Schmidt spaces on $L$ and $R$.
In what follows $A$ and $B$ will come from these two bases respectively.
For example, in the case of qubits these basis could be chosen to be the standard Pauli basis.
From the definition of a $0$-wall we have,
\begin{align}
  0 &= \frac{1}{2} \|[B,\mathrm{Ad}_U A]\|_F^2 \\
  &= \tr[U A A^{\dagger} U^\dagger B B^{\dagger}] - \tr[U A U^\dagger B U A^{\dagger} U^\dagger B^{\dagger}]
  \text{.}
  \nonumber
\end{align}
The first term is a kind of normalisation term and the second term is an out-of-time-order correlator.
We take a singular value decomposition of $U$ between $L \otimes L^\ast$ and $R \otimes R^\ast$ with singular values $\sigma$.
This can be represented diagrammatically as follows,
\begin{align}
  \includegraphics[width=0.8\linewidth]{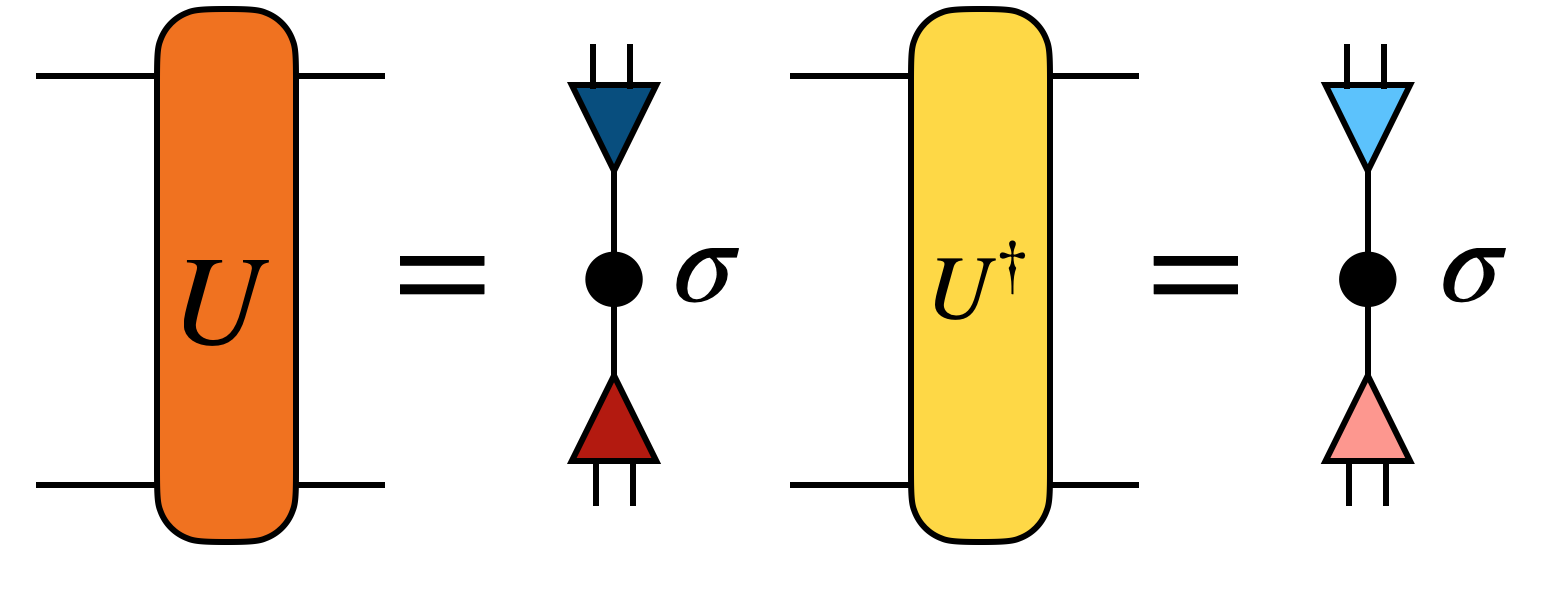}
  \nonumber
\end{align}
where the triangles represent the unitary isometries and the black dot represents the singular values.
The blue triangle pair are adjoints of one another and also horizontally flipped to avoid having to cross lines in diagrams.
By integrating over the two bases, one produces diagrams involving partial traces because this creates a resolution of the identity on Hilbert-Schmidt space.
Putting this together, for the out-of-time-order contribution,
\begin{align}
  \sum_{A,B}& \tr\left[ B U A U^\dagger B^\dagger U A U^\dagger \right] \\
  &= \vcenter{\hbox{\includegraphics[scale=0.35]{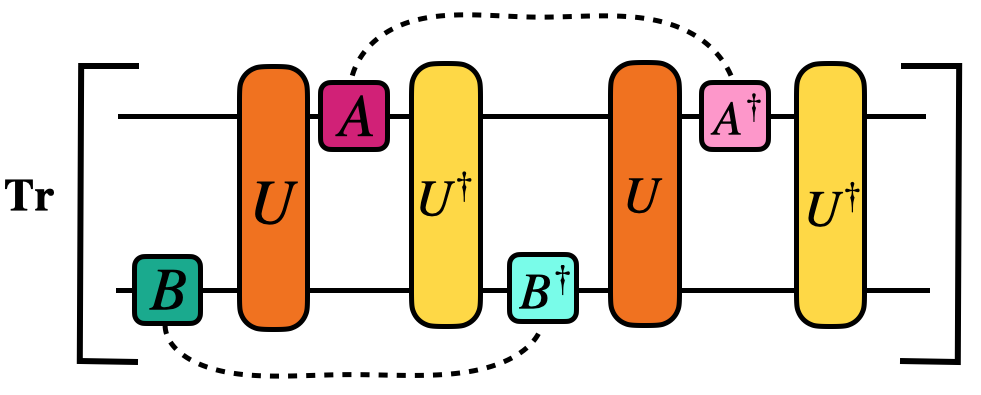}}} \\
  &= \vcenter{\hbox to 0.75\linewidth{\hfil\includegraphics[scale=0.35]{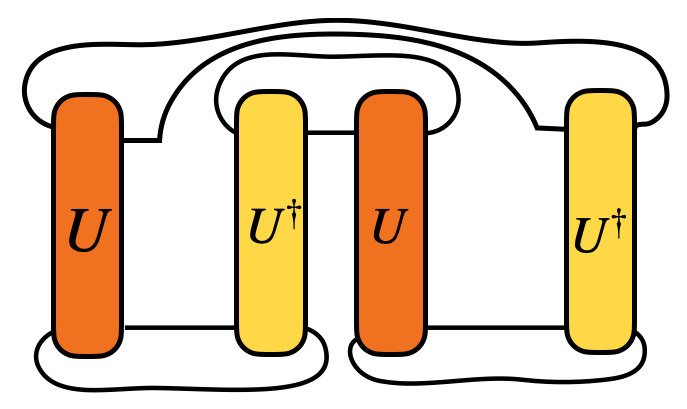}\hfil}} \\
  &= \vcenter{\hbox{\includegraphics[scale=0.35]{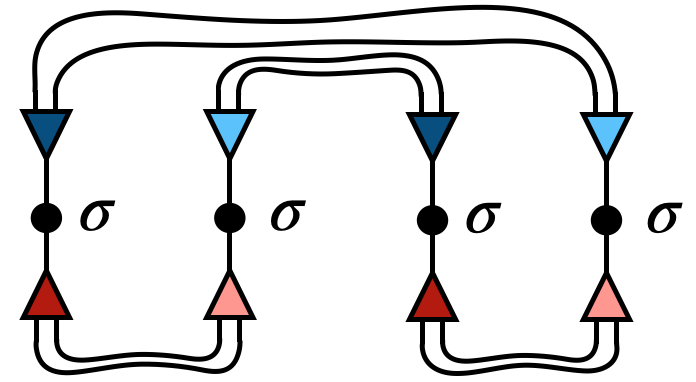}}}
  = \sum_i \sigma_i^4
\end{align}
and then for the time-ordered contribution,
\begin{align}
  \sum_{A,B}& \tr\left[ B B^\dagger U A A^\dagger U^\dagger \right] \\
  &= \vcenter{\hbox{\includegraphics[scale=0.35]{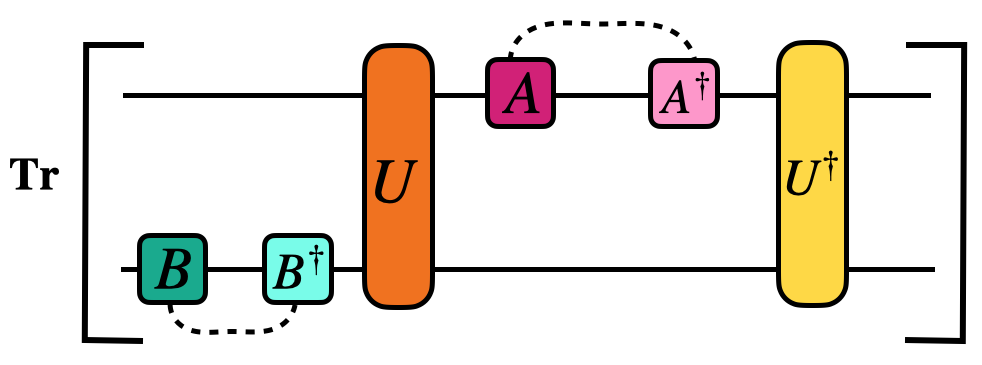}}} \\
  &= \vcenter{\hbox{\includegraphics[scale=0.35]{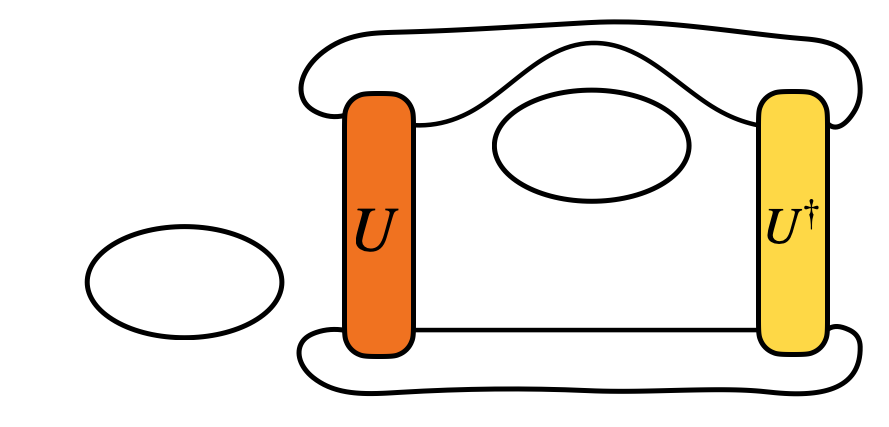}}} \\
  &= \vcenter{\hbox{\includegraphics[scale=0.35]{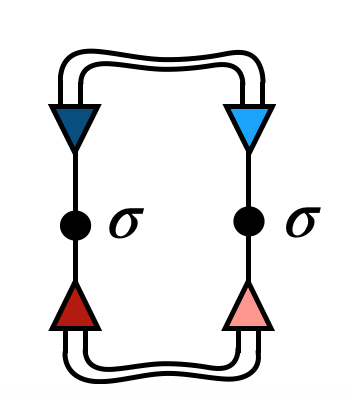}}}
  \vcenter{\hbox{\includegraphics[scale=0.35]{sigmas_2c.png}}}
  = \left( \sum_i \sigma_i^2 \right)^2
\end{align}
where for the final line we use the fact that we could contract the two unitaries to make a diagram consisting only of four loops.
Hence, the two loops in the original diagram is equal to the other portion of the diagram.
Combining the time-ordered and out-of-time-order contributions we see that,
\begin{align}
  0 = \left (\sum_i \sigma_i^2\right)^2 - \sum_i \sigma_i^4
  \text{.}
\end{align}
Hence, there is only one non-zero singular value.
This occurs exactly when $U$ is a product unitary over $L$ and $R$.
Essentially, in this calculation, we have computed an analogue to the entanglement R\'enyi 2-entropy for $U$ and find it to be unentangled.

\raggedbottom

\end{document}